\documentclass[10pt,oneside,a4paper]{article}

\usepackage{indentfirst}
\usepackage{amsmath}
\usepackage{amsfonts}
\usepackage{amsopn}
\usepackage{amsthm}
\usepackage{amssymb}
\usepackage{graphicx}
\usepackage[ruled,vlined]{algorithm2e}
\usepackage{maplestd2e}
\usepackage[dvipdfm,
            colorlinks=true,
            pagebackref=true,
            pdfstartview=FitH,
            pdfpagemode=none
            ]{hyperref}

\DeclareMathOperator{\arccot}{arccot}

\newtheorem{proposition}{Proposition}[section]

\newtheorem{theorem}{Theorem}[section]

\newtheorem{condition}{Condition}[section]

\numberwithin{equation}{section}

\begin{document}

\title{
Spherical Distribution of 5 Points with Maximal Distance Sum
\footnote{
Partially supported by a National Key Basic Research Project of China
(2004CB318000) and by National Natural Science Foundation of China
(10571095)
}
}

\author{
Xiaorong Hou\footnote{The author to whom all correspondence should be sent.}, Junwei Shao\\
College of Automation,\\
University of Electronic Science and Technology of China, Sichuan, PRC\\
E-mail: \href{mailto:houxr@uestc.edu.cn}{houxr@uestc.edu.cn},
\href{mailto:junweishao@gmail.com}{junweishao@gmail.com} }

\date{}

\maketitle


\begin{abstract}
In this paper, we mainly consider the problem of spherical distribution
of 5 points, that is, how to configure 5 points on a sphere such that the
mutual distance sum attains the maximum.
It is conjectured that the sum of distances is maximal if 5 points
form a bipyramid configuration in which case two points are positioned
at two poles of the sphere and the other three are positioned uniformly
on the equator.
We study this problem using interval methods and related technics,
and give a proof for the conjecture through computers in finite time.
\end{abstract}

\section{Introduction}
Studies on the problem of optimally arranging points on a sphere
can date back to over one hundred years ago,
when Thomson attempted to explain the periodic table in
terms of the ``plum pudding'' model of the atom.
Since then, several varied problems were proposed,
and some of such problems are still unsolved now \cite{croft:1}.
In general, these problems involve finding configurations
of points on the surface of a sphere that maximize or minimize
some given quantities, some of them are directly relevant to physics or
chemistry where stable configurations tend to minimize some form of
energy expression.

The problem has the following general form.
Let $x_1, x_2, \ldots, x_n$ be points on the unit sphere $S^{m-1}$ of
the Euclidean space $\mathbb{R}^m$, denote
\begin{equation} \label{vxn}
V(X_n,m,\lambda) =
\sum\limits_{1 \leq i < j \leq n} \left| x_i - x_j \right|^\lambda,
\end{equation}
where $X_n=(x_1,x_2,\ldots,x_n)$, and $\left| x_i-x_j \right|$ denotes
the Euclidean distance between $x_i$ and $x_j$.

For $\lambda \leq 0$, denote
\begin{equation} \label{vnn}
V_1(n,m,\lambda) = \min\limits_{X_n \subset S^{m-1}} V(X,n,m,\lambda)
\end{equation}
where
\begin{equation} \label{vn0}
V_1(n,m,0)=
\min\limits_{X_n \subset S^{m-1}} \sum\limits_{1 \leq i < j \leq n}
\log \frac{1}{\left| x_i - x_j \right|}
\end{equation}

When $m=3$, this is the 7th Problem listed by Steve Smale in
\textit{Mathematical Problems for the Next Century}
\cite{smale:1, smale:2}.

For $\lambda > 0$, denote
\begin{equation} \label{vnp}
V_2(n,m,\lambda) = \max\limits_{X_n \subset S^{m-1}} V(X,n,m,\lambda)
\end{equation}

So far as we know, G. P\'{o}lya and G. Szeg\"{o} \cite{polya:1} first studied problems
of such types in 1930s, since then, a number of results about $V_2(n,m,\lambda)$ have been derived.
For example,
L. Fejes T\'{o}th proved results for cases when $m=2, \lambda=1$
and when $n=m+1, \lambda=1$ \cite{toth:1}.
E. Hille considered the asymptotic properties of
$V_2(n,m,\lambda)/N$ when $n \to \infty$ for definite $m$ and $\lambda$,
and gave some results \cite{hille:1}.
K. B. Stolarsky proved bounds of $V_2(n,m,\lambda)$ for
definite $m$ and $\lambda$ in \cite{stolarsky:1, stolarsky:2},
and gave some properties of point distributions corresponding $V_2(n,m,\lambda)$
when $m = 2$ and $m=3$
in \cite{stolarsky:3, stolarsky:4, stolarsky:5}.
R. Alexander also proved bounds of $V_2(n,3,1)$ in
\cite{alexander:1}, and discussed some generalized
sums of distances in \cite{alexander:2, alexander:3}.
G. D. Chakerian and M. S. Klamkin proved bounds of $V_2(n,m,1)$ in
\cite{chakerian:1}.
J. Berman and K. Hanes proved a property of the point distribution
corresponding $V_2(n,3,1)$, and deduced some numerical results
in \cite{berman:1}.
G. Harman, J. Beck, T. Amdeberhan proved bounds of $V_2(n,m,\lambda)$ in
\cite{harman:1, beck:1, amdeberhan:1}.
Similar problems were also discussed in
\cite{furedi:1, ali:1, saff:1, minghuijiang:1}.

For $V_2(5,3,1)$, numerical computations show
evidences for the conjecture that,
it is obtained when 5 points form a bipyramid configuration
in which case two points are on the two poles of $S^2$, while three other
points are uniformly distributed on the equator.
In this paper, we study this problem via interval arithmetic,
and prove the conjecture through computer in comparatively short time.
For related problems, this guides a different method.

The main ideal of our proof is as follows.
Firstly we express $V(X_5,3,1)$ as a function under certain coordinate system,
secondly we exclude a domain where the bipyramid configuration
is proved to correspond an only maximum of $V(X_5,3,1)$,
lastly we subdivide the remaining domain,
and prove that function values in these subdomains are less then
the previous maximum obtained. So we complete the proof of the conjecture.

\section{Mathematical descriptions of the problem} \label{5points:math}

\subsection{Spherical coordinate system}
We choose the spherical coordinate system as showed in Fig. \ref{coordsys}.
A point $P$ on $S^2$ is identified by $(1,\phi,\theta)$,
where $\phi \in \left[-\frac{\pi}{2},\frac{\pi}{2}\right]$
is the angle from vector $\overrightarrow{OH}$,
i.e., the projection of vector $\overrightarrow{OP}$ in $xoy$-plane,
to vector $\overrightarrow{OP}$,
positive if the $z$-coordinate of $P$ is positive,
and $\theta \in [-\pi, \pi)$ is the angle from $x$-axis to
vector $\overrightarrow{OH}$, positive if the $y$-coordinate of $P$ is positive.

\begin{figure}[htbp]
\centering
\includegraphics[width=0.60\textwidth]{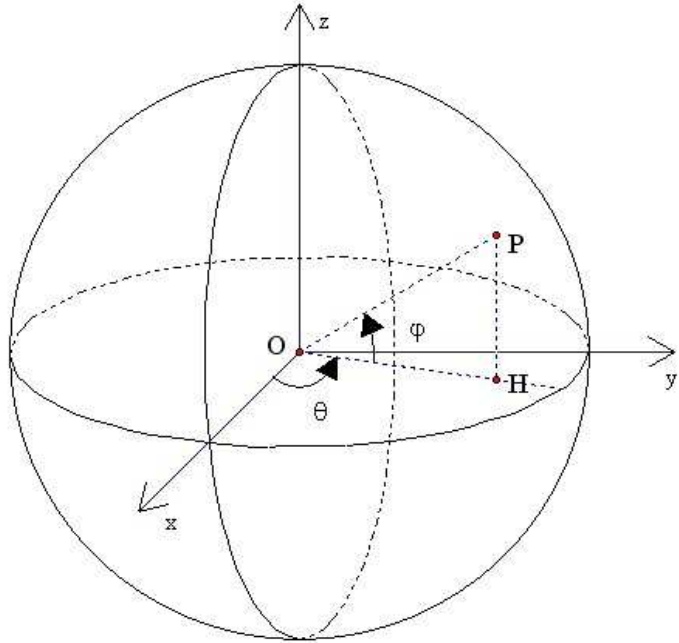}
\caption{The spherical coordinate system}
\label{coordsys}
\end{figure}

According to such definitions, we have following
formulas transforming from
spherical coordinate $(1,\phi,\theta)$ to Cartesian coordinate $(x,y,z)$,
\begin{equation}
\left\{
\begin{array}{rcl}
x & = & \cos(\phi) \cos(\theta), \\
y & = & \cos(\phi) \sin(\theta), \\
z & = & \sin(\phi).
\end{array}
\right.
\end{equation}

Considering the spherical symmetry, we can choose the spherical coordinates
for 5 points as follows:
\begin{equation}
A(1,0,0),B(1,\phi_1,\pi),C(1,\phi_2,\theta_2),D(1,\phi_3,\theta_3),E(1,\phi_4,\theta_4).
\end{equation}
Thus the sum of mutual distances of these points is
\begin{equation}
\begin{split}
& f(\phi_1, \phi_2, \theta_2, \phi_3, \theta_3, \phi_4, \theta_4) \\
= & \sqrt {2+ 2\,\cos \left( \phi_{{1}} \right)}+\sqrt {2-2\,\cos \left(
\phi_{{2}} \right) \cos \left( \theta_{{2}} \right) }\\
&+\sqrt {2-2\,\cos\left( \phi_{{3}} \right) \cos \left( \theta_{{3}} \right) }+\sqrt {2
-2\,\cos \left( \phi_{{4}} \right) \cos \left( \theta_{{4}} \right) }\\
&+\sqrt {2\,\cos \left( \phi_{{1}} \right) \cos \left( \phi_{{2}}
 \right) \cos \left( \theta_{{2}} \right) +2-2\,\sin \left( \phi_{{1}}
 \right) \sin \left( \phi_{{2}} \right) }\\
&+\sqrt {2\,\cos \left( \phi_{{1}} \right) \cos \left( \phi_{{3}} \right) \cos
\left( \theta_{{3}}
 \right) +2-2\,\sin \left( \phi_{{1}} \right) \sin \left( \phi_{{3}}
 \right) }\\
& +\sqrt {2\,\cos \left( \phi_{{1}} \right) \cos \left( \phi_{
{4}} \right) \cos \left( \theta_{{4}} \right) +2-2\,\sin \left( \phi_{
{1}} \right) \sin \left( \phi_{{4}} \right) }\\
&+\sqrt {-2\,\cos \left( \phi_{{3}} \right) \cos \left( \phi_{{2}} \right)\cos \left(
\theta_{{2}}-\theta_{{3}} \right)  +2-2\,\sin \left( \phi_{{2}} \right) \sin
 \left( \phi_{{3}} \right) }\\
&+\sqrt {-2\,\cos \left( \phi_{{2}} \right) \cos \left( \phi_{
{4}} \right)\cos \left( \theta_{{2}}-
\theta_{{4}} \right)  +2-2\,\sin \left( \phi_{{2}} \right) \sin \left( \phi_{{4
}} \right) }\\
&+\sqrt {-2\,\cos \left( \phi_{{3}} \right) \cos \left( \phi_{{4}} \right)
\cos \left( \theta_{{3}}-\theta_{{4}} \right) +2-2\,
\sin \left( \phi_{{3}} \right) \sin \left( \phi_{{4}} \right) }.
\end{split}
\end{equation}

\subsection{Bipyramid distribution} \label{sec:bipyramid}
Spherical coordinates of 5 points corresponding a bipyramid distribution are not unique,
but the following 5 points indeed form a bipyramid configuration,

\begin{equation} \label{bipycoords}
A(1,0,0),B(1,-\dfrac{\pi}{3},\pi),C(1,\dfrac{\pi}{3},\pi),D(1,0,-\dfrac{\pi}{2}),
E(1,0,\dfrac{\pi}{2}),
\end{equation}
as showed in Fig. \ref{bipydistrib}.

\begin{figure}[htbp]
\centering
\includegraphics[width=0.60\textwidth]{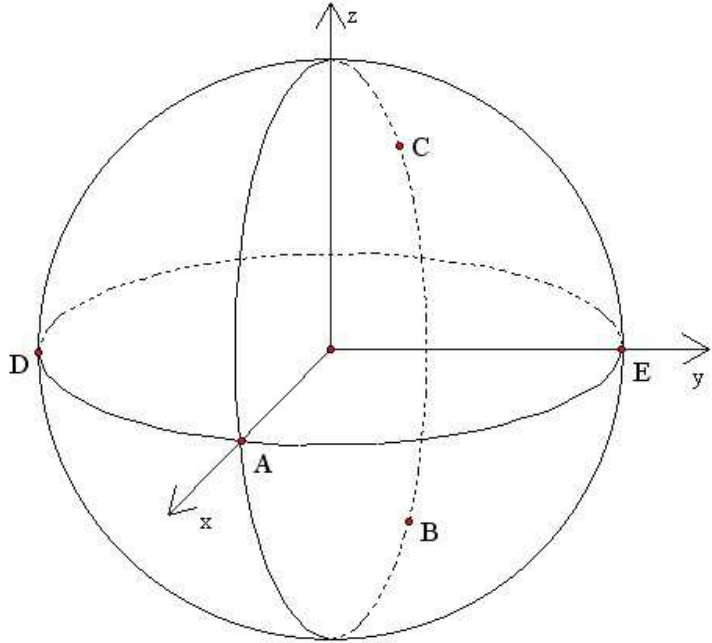}
\caption{The bipyramid distribution}
\label{bipydistrib}
\end{figure}

Denote the corresponding values of
$(\phi_1, \phi_2, \theta_2, \phi_3, \theta_3, \phi_4, \theta_4)$ by
$$
\Theta_{bp} = (-\dfrac{\pi}{3}, \dfrac{\pi}{3}, \pi, 0, -\dfrac{\pi}{2}, 0, \dfrac{\pi}{2}),
$$
then the corresponding value of function $f$ is
\begin{equation} \label{fmax}
\begin{split}
fmax & = f(\Theta_{bp}) \\
& = 3\,\sqrt {3}+6\,\sqrt {2}+2\\
& \approx 15.68143380,
\end{split}
\end{equation}
and the Hessian matrix of $f$ is
{\footnotesize
\begin{equation}
\left( \begin {array}{ccccccc}
{\dfrac {  -\sqrt {3}  }
{  2  }}&{\dfrac {  \sqrt {3}  }{
4  }}&0&{\dfrac {  -\sqrt {2}  }{
4  }}&{\dfrac {  \sqrt {6}  }{
4  }}&{\dfrac {  -\sqrt {2}  }{  4
}}&{\dfrac {  -\sqrt {6}  }{  4
}}\\\noalign{\medskip}{\dfrac {  \sqrt {3}  }{
4  }}&{\dfrac {  -\sqrt {3}  }{
2  }}&0&{\dfrac {  -\sqrt {2}  }{
4  }}&{\dfrac {  -\sqrt {6}  }{
4  }}&{\dfrac {  -\sqrt {2}  }{
4  }}&{\dfrac {  \sqrt {6}  }{
4  }}\\\noalign{\medskip}0&0&{\dfrac {  -2\,\sqrt {3}-3
\,\sqrt {2}  }{  24  }}&{\dfrac {  -
\sqrt {6}  }{  16  }}&{\dfrac {  \sqrt {2
}  }{  16  }}&{\dfrac {  \sqrt {6}
}{  16  }}&{\dfrac {  \sqrt {2}
}{  16  }}\\\noalign{\medskip}{\dfrac {
-\sqrt {2}  }{  4  }}&{\dfrac {
-\sqrt {2}  }{  4  }}&{\dfrac {
-\sqrt {6}  }{  16  }}&{\dfrac {
-3\,\sqrt {2}-4  }{  8  }}&0&{\dfrac {
-1  }{  2  }}&0\\\noalign{\medskip}{
\dfrac {  \sqrt {6}  }{  4  }}&{\dfrac {
-\sqrt {6}  }{  4  }}&{\dfrac {
\sqrt {2}  }{  16  }}&0&{\dfrac {
-3\,\sqrt {2}-4  }{  8  }}&0&{\dfrac {
1  }{  2  }}\\\noalign{\medskip}{
\dfrac {  -\sqrt {2}  }{  4  }}&{\dfrac {
-\sqrt {2}  }{  4  }}&{\dfrac {
\sqrt {6}  }{  16  }}&{\dfrac {
-1  }{  2  }}&0&{\dfrac {  -3\,
\sqrt {2}-4  }{  8  }}&0\\\noalign{\medskip}{
\dfrac {  -\sqrt {6}  }{  4  }}&{\dfrac {
\sqrt {6}  }{  4  }}&{\dfrac {
\sqrt {2}  }{  16  }}&0&{\dfrac {  1
}{  2  }}&0&{\dfrac {  -3\,\sqrt {2}-4
}{  8  }}\end {array} \right).
\end{equation}
}
This matrix is negative definite, so
the bipyramid distribution corresponds a maximum
of function $f$.

\subsection{Inequality form}
As a matter of fact, what we are to prove is the following inequality,

\begin{equation} \label{inequality}
f(\phi_1, \phi_2, \theta_2, \phi_3, \theta_3, \phi_4, \theta_4) \leq fmax, \quad
(\phi_1, \phi_2, \theta_2, \phi_3, \theta_3, \phi_4, \theta_4) \in \mathcal{D}.
\end{equation}
where
$$
\mathcal{D}=
\left(
[-\frac{\pi}{2},\frac{\pi}{2}],
[-\frac{\pi}{2},\frac{\pi}{2}], [-\pi, \pi),
[-\frac{\pi}{2},\frac{\pi}{2}], [-\pi, \pi),
[-\frac{\pi}{2},\frac{\pi}{2}], [-\pi, \pi)
\right),
$$
and the equality holds if and only if
$(\phi_1, \phi_2, \theta_2, \phi_3, \theta_3, \phi_4, \theta_4)=\Theta_{bp}$.

In the remaining part of this paper, we will according to following
steps to prove this inequality.

\begin{enumerate}
\item
Giving some restricted conditions
and results to demonstrate that we only need to prove the inequality
over a subdomain of $\mathcal{D}$, i.e., $\mathcal{D}^{(1)} \cup \mathcal{D}^{(2)}$
(see Eq. \eqref{eqn:bipyramid}).

\item
Analyzing interval Hessian matrices
(Theorem \ref{th:pstv}, \ref{th:posdefeig} and \ref{th:extremepoint})
to prove that the equality holds only at $\Theta_{bp}$ over a subdomain of
$\mathcal{D}^{(1)} \cup \mathcal{D}^{(2)}$,
i.e., $\mathcal{D}_{bp}$ (see Proposition \ref{prop:bipyramid}).

\item
Analyzing interval Hessian matrices
(Theorem \ref{th:nonpos} and \ref{th:notextremepoint})
to prove the corresponding strict inequality
holds over a subdomain of
$\mathcal{D}^{(1)} \cup \mathcal{D}^{(2)}$,
i.e., $\mathcal{D}_{p}$ (see Proposition \ref{prop:pyramid}).

\item
Making use of the interval arithmetic(\S\,\ref{intarith}) to prove the corresponding strict inequality
holds over the remaining domains, i.e.,
$(\mathcal{D}^{(1)} \cup \mathcal{D}^{(2)}) \backslash (\mathcal{D}_{bp} \cup \mathcal{D}_{p})$
(see Eq. \eqref{strictinequality}).
\end{enumerate}

\section{Restricted conditions and verification domain}
\subsection{Some results}
What we are to prove is in fact that, there exists no distribution
of 5 points exclude the bipyramid distribution
corresponding larger distance sum then $fmax$.
We need following results so as to simplify this problem.
\begin{proposition} \label{prop:phi1}
If some configuration of 5 points corresponds larger function value of $f$
then the bipyramid configuration,
and $AB$ is the second largest distance in ${5 \choose 2} = 10$ distances,
$\phi_1$ should satisfy
\begin{equation}
\phi_1 \geq -2\,\arccos ( \sqrt {3}/6 +\sqrt {2}/3 )
\end{equation}
\end{proposition}
\begin{proof}
From Equation \eqref{fmax}, we know that in order to attain larger distance
sum than that the bipyramid configuration corresponds,
the second largest distance must be not smaller then
$$
( ( 3\,\sqrt{3} + 6\,\sqrt{2} + 2 ) - 2 ) /9 =
\sqrt {3}/3 + 2\,\sqrt {2}/3.
$$
With the condition that $AB$ is the second largest distance,
the result required can be deduced immediately.
\end{proof}

\begin{proposition} \label{prop:halfsphere}
If 5 points are on the same half sphere,
$f$ can not attains its maximum.
\end{proposition}
\begin{proof}
Without loss of generality,
suppose $z$-coordinates of 5 points are all nonpositive,
if the $z$-coordinate of some point is negative,
we move it to the symmetric position with respect to the $xoy$-plane,
then we will get a larger distance sum.

If 5 points are all distributed on the $xoy$-plane,
the maximal distance sum is \cite{toth:1} (5 points form a regular pentagon)
$5\,\cot \frac{\pi}{10}$,
which is obviously smaller then the mutual distance
sum corresponding the bipyramid configuration
(see \S\,\ref{sec:bipyramid}).
\end{proof}

\begin{proposition} \label{prop:parder}
If a partial derivative of function $f$ does not vary signs
in a domain, then there exists no stationary point of $f$ in
this domain.
\end{proposition}

\begin{theorem} \cite{berman:1} \label{th:gradient}
Let $p_1,\ldots,p_n$ be points on the unit sphere $S^2$ in $\mathbb{R}^3$.
Let $f:\, S^2 \to \mathbb{R}$ be defined by
$f(x)=\sum\limits_{i=1}^n \left| x-p_i \right|$.
if $f$ has a maximum at $p$, then $p=q/|q|$, where
$q=\sum\limits_{i=1}^n (p-p_i)/\left| p-p_i \right|$.
\end{theorem}

\begin{theorem} \cite{stolarsky:1} \label{th:mindis}
Suppose the 5 points are placed so that
function $f$ is maximal, then any distance between two points
cannot be less then $\frac{2}{15}$.
\end{theorem}

\subsection{Some restricted conditions} \label{conditions}

We can consider the problem under following restricted conditions
due to above results.

\begin{condition} \label{as:secondlength}
$AB$ is the second largest distance in ${5 \choose 2} = 10$ distances.
\end{condition}

\begin{condition} \label{as:pointc}
$D$ is on the left half sphere,
$C,E$ are on the right half sphere, $C$ is above $E$.
\end{condition}

\begin{condition}[by Proposition \ref{prop:phi1}] \label{as:phi1}
$\phi_1 \geq -2\,\arccos ( \sqrt {3}/6 +\sqrt {2}/3 )$
\end{condition}

\begin{condition}[by Proposition \ref{prop:halfsphere}] \label{as:halfsphere}
Five points are not on any half sphere.
\end{condition}

\begin{condition}[by Theorem \ref{th:mindis}] \label{as:mindis}
Distances between any two points are larger then $\frac{2}{15}$.
\end{condition}

\subsection{Domain subdivision}

Under these conditions, the bipyramid configuration
(corresponding the maximal distance sum conjectured)
and the pyramid configuration (corresponding another stationary point of
the function $f$) each corresponds only one coordinate representation.
Further more,
we can divide the domain in which we need to verify no distribution
of points corresponds a larger distance sum into the following two
subdomains:
\begin{enumerate}
\item
$D$ is on the upper half sphere (denote this domain by $\mathcal{D}^{(1)}$):

$
\begin{array}{l}
\phi_{{1}}\in[-2\,\arccos \left( \sqrt {3}/6 +\sqrt {2}/3 \right),0],\\
\phi_{{2}}\in[-\pi/2,0],\\
\theta_{{2}}\in[0,\pi],\\
\phi_{{3}}\in[0,\pi/2],\\
\theta_{{3}}\in[-\pi,0],\\
\phi_{{4}}\in[-\pi/2,0],\\
\theta_{{4}}\in[0,\pi ].
\end{array}
$

\item
$D$ is on the lower half sphere, $C$ is on the
upper half sphere (denote this domain by $\mathcal{D}^{(2)}$):

$
\begin{array}{l}
\phi_{{1}}\in[-2\,\arccos \left( \sqrt{3}/6 + \sqrt{2}/3 \right),0],\\
\phi_{{2}}\in[0,\pi/2],\\
\theta_{{2}}\in[0,\pi],\\
\phi_{{3}}\in[-\pi/2,\pi/2],\\
\theta_{{3}}\in[-\pi,0],\\
\phi_{{4}}\in[-\pi/2,\pi/2],\\
\theta_{{4}}\in[0,\pi].
\end{array}
$

\end{enumerate}

Now, we are to prove that, under Condition
\ref{as:secondlength} - \ref{as:mindis},
function $f$ attains its maximum in
$\mathcal{D}^{(1)}$ and $\mathcal{D}^{(2)}$
at the only point
corresponding the bipyramid distribution of $A,B,C,D,E$, i.e.,

\begin{equation} \label{eqn:bipyramid}
f(\phi_1, \phi_2, \theta_2, \phi_3, \theta_3, \phi_4, \theta_4) \leq fmax, \quad
(\phi_1, \phi_2, \theta_2, \phi_3, \theta_3, \phi_4, \theta_4) \in \mathcal{D}^{(1)} \cup \mathcal{D}^{(2)}.
\end{equation}
where the equality holds if and only if
$(\phi_1, \phi_2, \theta_2, \phi_3, \theta_3, \phi_4, \theta_4) = \Theta_{bp}$.

In the following parts of this paper,
we will illustrate the domain verification methods,
and detailed steps as well as results.

\section{Domain near coordinates corresponding the bipyramid distribution}

\subsection{Interval methods}
We first briefly introduce the interval methods we used in our proof.

\subsubsection{Interval arithmetic} \label{intarith}
We define an interval as a set \cite{volker:1}:
\begin{equation}
X=[a,b]=\{ x: a \leq x \leq b \},
\end{equation}
where $a,b \in \mathbb{R}$.
$\underline{X}, \overline{X}$ respectively
denote the left and right vertexes of the interval $X$.

For intervals $X$ and $Y$, if $x > y$ for each $x \in X$ and each $y \in Y$,
we say that $X > Y$. Other interval relations are understood the same way.
An $n$-dimensional ``interval vector'' is an $n$-tuple of intervals
$\mathbf{X}=(X_1,\ldots,X_n)$, which is used to denote some rectangular domain
in $\mathbb{R}^n$. Let $\mathbb{IR}$ be the set of intervals over $\mathbb{R}$,
and $\mathbb{IR}^n$ be the set of $n$-dimensional ``interval vectors''.

We can define an imbedding from $\mathbb{R}$ to $\mathbb{IR}$ as follows
$$\mu(x)=[x,x],$$
thus for numbers in $\mathbb{R}$, we can also consider them as intervals.

We define interval arithmetic over  $\mathbb{IR}$ as
$$X \circ Y = \{x \circ y: x \in X, y \in Y \},$$
where $\circ$ is $``+", ``-", ``*"$ or $``/"$.
Further more, for an elementary function $f$, we define a
corresponding elementary mapping as
$$f(X) = \{ f(x): x \in X \}.$$

When operands of interval arithmetic or arguments of
elementary functions are intervals, we consider
underlying computations are interval computations defined above,
and the interval computation is of the same precedence as the corresponding
arithmetic computation.

Under above definitions, an arbitrary elementary function
$f:\mathbb{R}^n \to \mathbb{R}$ can be expanded to a mapping over
$\mathbb{IR}^n \to \mathbb{IR}$:
\begin{equation}
\tilde{f}(\mathbf{X}) =f(\mathbf{X}).
\end{equation}
Through such $\tilde{f}$,
We can get an interval which contains the function range of $f$ over rectangular
domain $\mathbf{X}$, this is the critical point we solve the problem.

As a matter of fact, there are related programs used to process interval arithmetic,
such as the procedure \textbf{evalr} can be used to implement interval arithmetic
without errors.
But in practice, it may be not necessary to implement errorless interval arithmetic,
because what we get from interval arithmetic are just intervals contains
the ranges of function values. Another problem is that acting such errorless interval arithmetic
is always time-consuming, thus it cannot meet our needs.

Considering the efficiency and the accuracy,
we wrote an interval arithmetic package \textbf{IntervalArithmetic} based on
the Maple system.
The package uses rational numbers as interval vertexes,
and acts computations with controllable errors.
In fact the result it computes for $f(\mathbf{X})$
is an larger interval containing $\tilde{f}(\mathbf{X})$,
and the difference can reduce to zero
as intervals of $\mathbf{X}$ shrink to points.
For the detailed code, see Appendix \ref{intervalarithmetic}.

\subsubsection{Interval matrices}
Relations of real matrices of the same order are understood componentwise.
An interval matrix is defined as the following set of matrices:
$$([\underline{a}_{ij},\overline{a}_{ij}])=[\underline{A},\overline{A}]
=\{A \in \mathbb{R}^{n\, \times \, n}: \underline{A} \leq A \leq \overline{A}\},$$
where
$$\underline{A}=(\underline{a}_{ij}),\overline{A}=(\overline{a}_{ij}).$$
When $\underline{A}$ and $\overline{A}$ are symmetric, we call the set of
symmetric matrices in $[\underline{A},\overline{A}]$ a symmetric interval matrix
which is also denoted by $[\underline{A},\overline{A}]$.

For a interval matrix $[\underline{A},\overline{A}]$, denote its midpoint matrix
by $A_c=\dfrac{\underline{A}+\overline{A}}{2}$, radius matrix by
$A_\delta=\dfrac{\overline{A}-\underline{A}}{2}$.
For a real symmetric matrix $A$,
it is well know that all its eigenvalues are real,
we denote them in decreasing order by
$\lambda_1(A) \geq \lambda_2(A) \geq \cdots \geq \lambda_n(A)$,
and denote the spectrum of $A$ (i.e. the maximum eigenvalue modulus)
by $\varrho(A)$.
For bounds of eigenvalues of matrices in
an interval matrix, it can be directly deduced from
the Wielandt-Hoffman theorem \cite{golub:1} that

\begin{theorem} \label{th:wielandt-hoffman}
For a symmetric interval matrix $[\underline{A},\overline{A}]$, the set
$$
\{ \lambda_i(A) : A \in [\underline{A},\overline{A}]\}
$$
is a compact interval, denote this compact interval by
$$
[\underline{\lambda}_i([\underline{A},\overline{A}]),
\overline{\lambda}_i([\underline{A},\overline{A}])],
\, 1 \leq i \leq n,
$$
then
$$
[\underline{\lambda}_i([\underline{A},\overline{A}]),
\overline{\lambda}_i([\underline{A},\overline{A}])]
\subseteq [\lambda_i(A_c)-\varrho(A_\delta),\lambda_i(A_c)+\varrho(A_\delta)],
\, i=1,\ldots,n.
$$
\end{theorem}

In fact,
$\overline{\lambda}_1([\underline{A},\overline{A}])$
and
$\underline{\lambda}_n([\underline{A},\overline{A}])$
can be solved explicitly \cite{rohn:1}, that is,

\begin{theorem} \label{th:exteig}
A real symmetric interval matrix
$$
([\underline{a}_{ij},\overline{a}_{ij}])
=\{A \in \mathbb{R}^{n\, \times \, n}: \underline{A} \leq A \leq \overline{A},\underline{A}
=(\underline{a}_{ij}),\overline{A}=(\overline{a}_{ij})\}
$$
corresponds following $2^{n-1}$ vertex matrices:
$$A_k=(a_{kij}), 0 \leq k \leq 2^{n-1}-1,$$
where we denote the binary representation for $k$
by $k=(k_1 k_2 \cdots k_n)_2$, and
$$
a_{kij}=\frac{1}{2}(\underline{a}_{ij}+\overline{a}_{ij}+(-1)^{k_i + k_j}
(\underline{a}_{ij}-\overline{a}_{ij})).
$$
For matrices in this symmetric interval matrix,
minimal (or maximal) eigenvalues of them attain the minimum (respectively, maximum)
at some vertex matrix $A_k$.
\end{theorem}

For a real symmetric interval matrix $[\underline{A},\overline{A}]$,
we say it is positive (semi)definite if $A$ is positive (semi)definite
for each $A \in [\underline{A},\overline{A}]$,
and it is nonpositive (semi)definite if $A$ is not positive (semi)definite
for each $A \in [\underline{A},\overline{A}]$.
Definitions such as negative (semi)definiteness, nonnegative (semi)definiteness of
$[\underline{A},\overline{A}]$ are understood the similar way.

Now we introduce the results for verifying
positive definiteness and nonpositive semidefiniteness
of symmetric interval matrices, which can directly deduce
criterions for negative definiteness, nonnegative definiteness, etc.

Rohn has given the following theorem \cite{rohn:1}, which is an
improvement on results in \cite{shi:1}, we state it in a
varied way that adapts to be understood as an algorithm.

\begin{theorem} \label{th:pstv}
The real symmetric interval matrix
$$([\underline{a}_{ij},\overline{a}_{ij}])
=\{A \in \mathbb{R}^{n\, \times \, n}: \underline{A} \leq A \leq \overline{A},\underline{A}
=(\underline{a}_{ij}),\overline{A}=(\overline{a}_{ij})\}$$
is positive definite if and only if the following
$2^{n-1}$ vertex matrices are all positive definite:
$$A_k=(a_{kij}), 0 \leq k \leq 2^{n-1}-1,$$
where we denote the binary representation for $k$ by
$k=(k_1 k_2 \cdots k_n)_2$, and
$$
a_{kij}=\frac{1}{2}(\underline{a}_{ij}+\overline{a}_{ij}+(-1)^{k_i + k_j}
(\underline{a}_{ij}-\overline{a}_{ij})).
$$
\end{theorem}

Theorem \ref{th:exteig} in fact implies Theorem \ref{th:pstv},
because the positive definiteness of a symmetric interval matrix
is equal that the minimum of minimal eigenvalues of matrices in it
is positive.
So we can get from Theorem \ref{th:wielandt-hoffman} a
sufficient condition for determining the positive definiteness
of a symmetric interval matrix.

\begin{theorem} \label{th:posdefeig}
The symmetric interval matrix $[\underline{A},\overline{A}]$
is positive definite if
$$
\lambda_n(A_c)-\varrho(A_\delta) > 0,
$$
where $A_c$ and $A_\delta$ are the midpoint matrix
and the radius matrix respectively.
\end{theorem}

Similarly, the nonpositive semidefiniteness of a symmetric interval matrix
is equal that the maximum of minimal eigenvalues of matrices in it
is negative, i.e.,

\begin{theorem} \label{th:nonpos}
The symmetric interval matrix $[\underline{A},\overline{A}]$
is nonpositive semidefinite if
$$
\lambda_n(A_c)+\varrho(A_\delta) < 0,
$$
where $A_c$ and $A_\delta$ are the midpoint matrix
and the radius matrix respectively.
\end{theorem}

The procedure \textbf{isdef} in the package \textbf{IntervalArithmetic}
can implement algorithms above to verify the properties of
symmetric interval matrices, such as positive definiteness,
negative definiteness and so on.

With the help of above theorems, we can use the following trivial results to
determine the extreme point of a function in a domain.

\begin{theorem} \label{th:extremepoint}
If $g \in C^2(D)$, $X_0 \in D$ is a stationary point of $g$,
and the Hessian matrix of $g$ over $D$ varies in a positive definite
real symmetric interval matrix, then $X_0$ is the minimum point of $g$
in $D$.
\end{theorem}

\begin{theorem} \label{th:notextremepoint}
If $g \in C^2(D)$, and the Hessian matrix of $g$ over $D$ varies in
a nonpositive semidefinite real symmetric interval matrix,
then no inner point of $D$ is the minimum point of $g$.
\end{theorem}

\subsection{Domain excluded near coordinates corresponding the bipyramid distribution}

Now we introduce a disturbance $\left[-\frac{\pi}{377},\frac{\pi}{377}\right]$
on coordinates corresponding the bipyramid distribution,
and obtain a rectangular domain, i.e.,
\begin{equation} \label{bpyva:pi}
\left( \begin {array}{c} \phi_{{1}}\\\noalign{\medskip}\phi_{{2}}
\\\noalign{\medskip}\theta_{{2}}\\\noalign{\medskip}\phi_{{3}}
\\\noalign{\medskip}\theta_{{3}}\\\noalign{\medskip}\phi_{{4}}
\\\noalign{\medskip}\theta_{{4}}\end {array} \right)
\in
\left( \begin {array}{c} \left[-{\frac {380}{1131}}\,\pi ,-{\frac {374}{
1131}}\,\pi \right]\\\noalign{\medskip}\left[{\frac {374}{1131}}\,\pi ,{\frac {
380}{1131}}\,\pi \right]\\\noalign{\medskip}\left[{\frac {376}{377}}\,\pi ,{
\frac {378}{377}}\,\pi \right]\\\noalign{\medskip}\left[-{\frac {1}{377}}\,\pi ,{
\frac {1}{377}}\,\pi \right]\\\noalign{\medskip}\left[-{\frac {379}{754}}\,\pi ,-
{\frac {375}{754}}\,\pi \right]\\\noalign{\medskip}\left[-{\frac {1}{377}}\,\pi ,
{\frac {1}{377}}\,\pi \right]\\\noalign{\medskip}\left[{\frac {375}{754}}\,\pi ,{
\frac {379}{754}}\,\pi \right]\end {array} \right).
\end{equation}
In this domain, $\theta_2$ varies in
$\left[{\frac {376}{377}}\,\pi ,{\frac {378}{377}}\,\pi \right]$,
which exceeds the bound we prescribed for $\theta_2$.
But since the periodicity of function $f$, it is of
no error. In fact, interval vertexes are represented by rational numbers
in the Maple package \textbf{IntervalArithmetic}, so these intervals
whose vertexes contain $\pi$ are enlarged to their
rational representations, that is, the rectangular domain
we actually obtain is
\begin{equation}
\mathcal{D}_{bp} =
\left( \begin {array}{c} \left[-{\frac {1055530689}{1000000000}},-{\frac {
1038864413}{1000000000}}\right]\\\noalign{\medskip}\left[{\frac {1038864413}{
1000000000}},{\frac {1055530689}{1000000000}}\right]\\\noalign{\medskip}\left[{
\frac {783314879}{250000000}},{\frac {98435181}{31250000}}\right]
\\\noalign{\medskip}\left[-{\frac {10416421265218147343}{
1250000000000000000000}},{\frac {10416421265218147343}{
1250000000000000000000}}\right]\\\noalign{\medskip}\left[-{\frac {315825893}{
200000000}},-{\frac {1562463189}{1000000000}}\right]\\\noalign{\medskip}\left[-{
\frac {10416421265218147343}{1250000000000000000000}},{\frac {
10416421265218147343}{1250000000000000000000}}\right]\\\noalign{\medskip}\left[{
\frac {1562463189}{1000000000}},{\frac {315825893}{200000000}}\right]
\end {array} \right).
\end{equation}

The interval Hessian matrix of $f$ over $\mathcal{D}_{bp}$ can be calculated
by interval arithmetic:
\begin{equation}
\mathcal{V} = \left(V_1,V_2,V_3,V_4,V_5,V_6,V_7\right),
\end{equation}
where $V_i(i=1,\ldots,7)$ are vectors as follows:\\[10pt]

{\tiny \raggedright
$
V_{{1}}= \left( \begin {array}{c} \left[-{\dfrac {9073071021}{10000000000}},
-{\dfrac {257887557}{312500000}}\right]\\\noalign{\medskip}\left[{\dfrac {
4158136493}{10000000000}},{\dfrac {1126402921}{2500000000}}\right]
\\\noalign{\medskip}\left[-{\dfrac {2503191333}{1000000000000}},{\dfrac {
2503191333}{1000000000000}}\right]\\\noalign{\medskip}\left[-{\dfrac {910889963}{
2500000000}},-{\dfrac {428510269}{1250000000}}\right]\\\noalign{\medskip}\left[{
\dfrac {6038013477}{10000000000}},{\dfrac {1552383121}{2500000000}}\right]
\\\noalign{\medskip}\left[-{\dfrac {72871197}{200000000}},-{\dfrac {685616431
}{2000000000}}\right]\\\noalign{\medskip}\left[-{\dfrac {3104766243}{5000000000}},
-{\dfrac {6038013477}{10000000000}}\right]\end {array} \right),
V_{{2}}= \left( \begin {array}{c} \left[{\dfrac {4158136493}{10000000000}},{
\dfrac {1126402921}{2500000000}}\right]\\\noalign{\medskip}\left[-{\dfrac {
2283824451}{2500000000}},-{\dfrac {8191611957}{10000000000}}\right]
\\\noalign{\medskip}\left[-{\dfrac {23364423}{781250000}},{\dfrac {747661531}
{25000000000}}\right]\\\noalign{\medskip}\left[-{\dfrac {114837637}{312500000}},-{
\dfrac {3397366803}{10000000000}}\right]\\\noalign{\medskip}\left[-{\dfrac {
1559063347}{2500000000}},-{\dfrac {240452637}{400000000}}\right]
\\\noalign{\medskip}\left[-{\dfrac {183740219}{500000000}},-{\dfrac {
424670851}{1250000000}}\right]\\\noalign{\medskip}\left[{\dfrac {6011315931}{
10000000000}},{\dfrac {779531673}{1250000000}}\right]\end {array} \right)
$
\\[10pt]
$
V_{{3}}= \left( \begin {array}{c} \left[-{\dfrac {2503191333}{1000000000000}
},{\dfrac {2503191333}{1000000000000}}\right]\\\noalign{\medskip}\left[-{\dfrac {
23364423}{781250000}},{\dfrac {747661531}{25000000000}}\right]
\\\noalign{\medskip}\left[-{\dfrac {3523502821}{10000000000}},-{\dfrac {
2901860369}{10000000000}}\right]\\\noalign{\medskip}\left[-{\dfrac {1628135751}{
10000000000}},-{\dfrac {359058763}{2500000000}}\right]\\\noalign{\medskip}\left[{
\dfrac {1556247929}{20000000000}},{\dfrac {9916175537}{100000000000}}\right]
\\\noalign{\medskip}\left[{\dfrac {718117527}{5000000000}},{\dfrac {
1628135749}{10000000000}}\right]\\\noalign{\medskip}\left[{\dfrac {7781239691}{
100000000000}},{\dfrac {2479043873}{25000000000}}\right]\end {array} \right),
V_{{4}}= \left( \begin {array}{c} \left[-{\dfrac {910889963}{2500000000}},-{
\dfrac {428510269}{1250000000}}\right]\\\noalign{\medskip}\left[-{\dfrac {114837637
}{312500000}},-{\dfrac {3397366803}{10000000000}}\right]\\\noalign{\medskip}\left[
-{\dfrac {1628135751}{10000000000}},-{\dfrac {359058763}{2500000000}}\right]
\\\noalign{\medskip}\left[-{\dfrac {214384921}{200000000}},-{\dfrac {
9891691177}{10000000000}}\right]\\\noalign{\medskip}\left[-{\dfrac {483991657}{
20000000000}},{\dfrac {97153283}{4000000000}}\right]\\\noalign{\medskip}\left[-{
\dfrac {10001389}{20000000}},-{\dfrac {4999305593}{10000000000}}\right]
\\\noalign{\medskip}\left[-{\dfrac {1041678267}{10000000000000}},{\dfrac {
1041678267}{10000000000000}}\right]\end {array} \right)
$
\\[10pt]
$
V_{{5}}= \left( \begin {array}{c} \left[{\dfrac {6038013477}{10000000000}},{
\dfrac {1552383121}{2500000000}}\right]\\\noalign{\medskip}\left[-{\dfrac {
1559063347}{2500000000}},-{\dfrac {240452637}{400000000}}\right]
\\\noalign{\medskip}\left[{\dfrac {1556247929}{20000000000}},{\dfrac {
9916175537}{100000000000}}\right]\\\noalign{\medskip}\left[-{\dfrac {483991657}{
20000000000}},{\dfrac {97153283}{4000000000}}\right]\\\noalign{\medskip}\left[-{
\dfrac {1058713931}{1000000000}},-{\dfrac {1002296947}{1000000000}}\right]
\\\noalign{\medskip}\left[-{\dfrac {1041678267}{10000000000000}},{\dfrac {
1041678267}{10000000000000}}\right]\\\noalign{\medskip}\left[{\dfrac {312434903}{
625000000}},{\dfrac {1250173619}{2500000000}}\right]\end {array} \right),
V_{{6}}= \left( \begin {array}{c} \left[-{\dfrac {72871197}{200000000}},-{
\dfrac {685616431}{2000000000}}\right]\\\noalign{\medskip}\left[-{\dfrac {183740219
}{500000000}},-{\dfrac {424670851}{1250000000}}\right]\\\noalign{\medskip}\left[{
\dfrac {718117527}{5000000000}},{\dfrac {1628135749}{10000000000}}\right]
\\\noalign{\medskip}\left[-{\dfrac {10001389}{20000000}},-{\dfrac {4999305593
}{10000000000}}\right]\\\noalign{\medskip}\left[-{\dfrac {1041678267}{
10000000000000}},{\dfrac {1041678267}{10000000000000}}\right]
\\\noalign{\medskip}\left[-{\dfrac {1071924603}{1000000000}},-{\dfrac {
2472922799}{2500000000}}\right]\\\noalign{\medskip}\left[-{\dfrac {607208011}{
25000000000}},{\dfrac {302494783}{12500000000}}\right]\end {array} \right)
$
\\[10pt]
$
V_{{7}}= \left( \begin {array}{c} \left[-{\dfrac {3104766243}{5000000000}},-
{\dfrac {6038013477}{10000000000}}\right]\\\noalign{\medskip}\left[{\dfrac {
6011315931}{10000000000}},{\dfrac {779531673}{1250000000}}\right]
\\\noalign{\medskip}\left[{\dfrac {7781239691}{100000000000}},{\dfrac {
2479043873}{25000000000}}\right]\\\noalign{\medskip}\left[-{\dfrac {1041678267}{
10000000000000}},{\dfrac {1041678267}{10000000000000}}\right]
\\\noalign{\medskip}\left[{\dfrac {312434903}{625000000}},{\dfrac {1250173619
}{2500000000}}\right]\\\noalign{\medskip}\left[-{\dfrac {607208011}{25000000000}},
{\dfrac {302494783}{12500000000}}\right]\\\noalign{\medskip}\left[-{\dfrac {
1058713929}{1000000000}},-{\dfrac {20045939}{20000000}}\right]\end {array}
\right)
$
\\[10pt]
}

Through Theorem \ref{th:pstv}, we can judge that
the symmetric interval matrix $\mathcal{V}$ is negative definite,
and by Theorem \ref{th:extremepoint}, the conjectured configuration indeed corresponds the
maximum of $f$ in $\mathcal{D}_{bp}$. That is

\begin{proposition} \label{prop:bipyramid}
The bipyramid distribution of 5 points represented by Eq. \eqref{bipycoords}
is the only maximal distance sum distribution in domain
$\mathcal{D}_{bp}$, i.e.,
\begin{equation} \label{eqn:bipyramid}
f(\phi_1, \phi_2, \theta_2, \phi_3, \theta_3, \phi_4, \theta_4) \leq fmax, \quad
(\phi_1, \phi_2, \theta_2, \phi_3, \theta_3, \phi_4, \theta_4) \in \mathcal{D}_{bp},
\end{equation}
where the equality holds if and only if
$(\phi_1, \phi_2, \theta_2, \phi_3, \theta_3, \phi_4, \theta_4)=\Theta_{bp}$.
\end{proposition}

\section{Domain near coordinates corresponding the pyramid distribution} \label{sec:pyramid}
Under conditions in \S\,\ref{conditions}, coordinates representing the pyramid distribution
are unique, while they corresponds a stationary point of function $f$,
and the function value on this point is too close to $fmax$,
therefore, we discuss it separately.

\subsection{Pyramid distribution}
The spherical coordinate corresponding the
pyramid distribution is

\begin{equation} \label{bycoords}
\begin{array}{l}
A(1,0,0),\\[5pt]
B\left(1,
-2\,\omega_1,
\pi\right),\\[5pt]
C\left(1,
\dfrac{\pi}{2} -\omega_1,
\pi\right),\\[5pt]
D\left(1,
\omega_2,
-\omega_3
\right),\\[5pt]
E\left(1,
\omega_2,
\omega_3
\right),
\end{array}
\end{equation}
where

{\footnotesize
\begin{equation}
\begin{split}
\omega_1 & = \arcsin \left( -\dfrac{3}{4} + \dfrac{\sqrt{2}}{2}+
\dfrac{\sqrt {41-28\,\sqrt {2}}} {4} \right),\\
\omega_2 & = -\arcsin \left(
\left( -\dfrac{3}{4} + \dfrac{\sqrt{2}}{2}+
\dfrac{\sqrt {41-28\,\sqrt {2}}} {4} \right)
\sqrt{1-
\left( -\dfrac{3}{4} + \dfrac{\sqrt{2}}{2}+
\dfrac{\sqrt {41-28\,\sqrt {2}}} {4} \right) ^2
}
\right),\\
\omega_3 & = \arccot \left( \frac{
\left( -\dfrac{3}{4} + \dfrac{\sqrt{2}}{2}+
\dfrac{\sqrt {41-28\,\sqrt {2}}} {4} \right) ^2
}
{
\sqrt{1-
\left( -\dfrac{3}{4} + \dfrac{\sqrt{2}}{2}+
\dfrac{\sqrt {41-28\,\sqrt {2}}} {4} \right) ^2
}
}
\right),
\end{split}\nonumber
\end{equation}
}\\
as showed in Fig. \ref{pyconfig}.

\begin{figure}[htbp]
\centering
\includegraphics[width=0.60\textwidth]{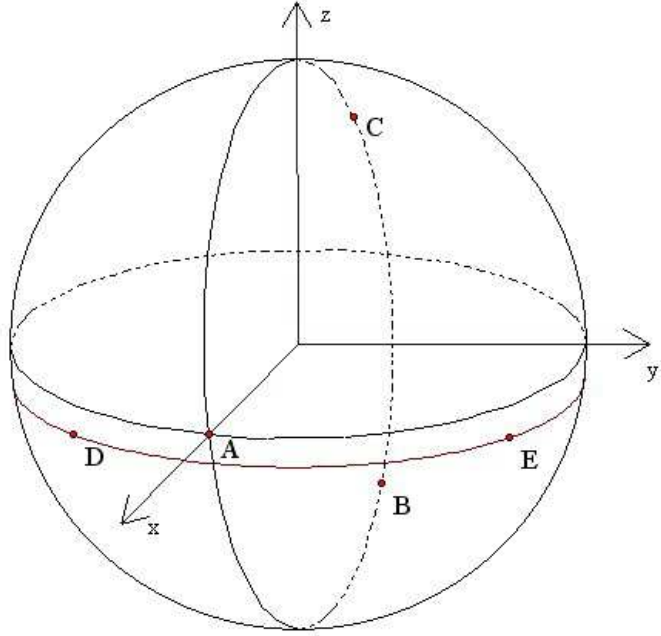}
\caption{The pyramid configuration} \label{pyconfig}
\end{figure}

Denote the corresponding values of
$(\phi_1, \phi_2, \theta_2, \phi_3, \theta_3, \phi_4, \theta_4)$
by
$$
\Theta_p = \left(
-2\,\omega_1,
\dfrac{\pi}{2} -\omega_1,
\pi,
\omega_2,
-\omega_3,
\omega_2,
\omega_3
\right),
$$
then the corresponding value of function $f$ is
\begin{equation}  \label{eqn:pyfmax}
\begin{split}
pyfmax & = f(\Theta_p)\\
& \approx 15.67482117.
\end{split}
\end{equation}

\subsection{Domain excluded near coordinates corresponding the pyramid distribution}
Similarly with the method we adopted near the bipyramid distribution,
we introduce a disturbance of $\left[-\frac{\pi}{791},\frac{\pi}{791}\right]$
on coordinates corresponding the pyramid distribution,
and finally obtain a rectangular domain
\begin{equation}
\mathcal{D}_{p} =
\left( \begin {array}{c} \left[-{\frac {5157880419}{10000000000}},-{\frac
{203137879}{400000000}}\right]\\\noalign{\medskip}\left[{\frac {1310916469}{
1000000000}},{\frac {263771963}{200000000}}\right]\\\noalign{\medskip}\left[{
\frac {156881049}{50000000}},{\frac {3145564327}{1000000000}}\right]
\\\noalign{\medskip}\left[-{\frac {39276321}{156250000}},-{\frac {
2434251099}{10000000000}}\right]\\\noalign{\medskip}\left[-{\frac {1508635943}{
1000000000}},-{\frac {375173149}{250000000}}\right]\\\noalign{\medskip}\left[-{
\frac {39276321}{156250000}},-{\frac {2434251099}{10000000000}}\right]
\\\noalign{\medskip}\left[{\frac {375173149}{250000000}},{\frac {1508635943
}{1000000000}}\right]\end {array} \right).
\end{equation}

The interval Hessian matrix of $f$ over $\mathcal{D}_{p}$ can be calculated
by interval arithmetic, i.e.,
\begin{equation}
\mathcal{W} = \left(W_1,W_2,W_3,W_4,W_5,W_6,W_7\right),
\end{equation}
where $W_i(i=1,\ldots,7)$ are vectors

{\tiny \raggedright
$
W_{{1}}= \left( \begin {array}{c} \left[-{\dfrac {68555671}{80000000}},-{
\dfrac {8085653887}{10000000000}}\right]\\\noalign{\medskip}\left[{\dfrac {
241020403}{625000000}},{\dfrac {4060471777}{10000000000}}\right]
\\\noalign{\medskip}\left[-{\dfrac {5396514777}{10000000000000}},{\dfrac {
5396514777}{10000000000000}}\right]\\\noalign{\medskip}\left[-{\dfrac {665453073}{
1000000000}},-{\dfrac {3261193771}{5000000000}}\right]\\\noalign{\medskip}\left[{
\dfrac {323160259}{1250000000}},{\dfrac {2727361243}{10000000000}}\right]
\\\noalign{\medskip}\left[-{\dfrac {6654530721}{10000000000}},-{\dfrac {
6522387553}{10000000000}}\right]\\\noalign{\medskip}\left[-{\dfrac {545472247}{
2000000000}},-{\dfrac {16158013}{62500000}}\right]\end {array} \right),
W_{{2}}= \left( \begin {array}{c} \left[{\dfrac {241020403}{625000000}},{
\dfrac {4060471777}{10000000000}}\right]\\\noalign{\medskip}\left[-{\dfrac {
567279649}{500000000}},-{\dfrac {544204613}{500000000}}\right]
\\\noalign{\medskip}\left[-{\dfrac {609173849}{50000000000}},{\dfrac {
1218347717}{100000000000}}\right]\\\noalign{\medskip}\left[-{\dfrac {68583739}{
400000000}},-{\dfrac {1584190467}{10000000000}}\right]\\\noalign{\medskip}\left[-{
\dfrac {1486861643}{2500000000}},-{\dfrac {5875949419}{10000000000}}\right]
\\\noalign{\medskip}\left[-{\dfrac {428648367}{2500000000}},-{\dfrac {
792095233}{5000000000}}\right]\\\noalign{\medskip}\left[{\dfrac {2937974709}{
5000000000}},{\dfrac {5947446571}{10000000000}}\right]\end {array} \right)
$
\\[10pt]
$
W_{{3}}= \left( \begin {array}{c} \left[-{\dfrac {5396514777}{10000000000000
}},{\dfrac {5396514777}{10000000000000}}\right]\\\noalign{\medskip}\left[-{\dfrac {
609173849}{50000000000}},{\dfrac {1218347717}{100000000000}}\right]
\\\noalign{\medskip}\left[-{\dfrac {8085854539}{100000000000}},-{\dfrac {
1541389841}{25000000000}}\right]\\\noalign{\medskip}\left[-{\dfrac {9959577971}{
100000000000}},-{\dfrac {9385595363}{100000000000}}\right]
\\\noalign{\medskip}\left[{\dfrac {116759761}{5000000000}},{\dfrac {109700733
}{4000000000}}\right]\\\noalign{\medskip}\left[{\dfrac {9385595361}{100000000000}}
,{\dfrac {9959577971}{100000000000}}\right]\\\noalign{\medskip}\left[{\dfrac {
2335195223}{100000000000}},{\dfrac {2742518293}{100000000000}}\right]
\end {array} \right),
W_{{4}}= \left( \begin {array}{c} \left[-{\dfrac {665453073}{1000000000}},-{
\dfrac {3261193771}{5000000000}}\right]\\\noalign{\medskip}\left[-{\dfrac {68583739
}{400000000}},-{\dfrac {1584190467}{10000000000}}\right]\\\noalign{\medskip}\left[
-{\dfrac {9959577971}{100000000000}},-{\dfrac {9385595363}{100000000000}
}\right]\\\noalign{\medskip}\left[-{\dfrac {8828398133}{10000000000}},-{\dfrac {
1678302277}{2000000000}}\right]\\\noalign{\medskip}\left[-{\dfrac {585736263}{
5000000000}},-{\dfrac {2143830967}{25000000000}}\right]\\\noalign{\medskip}\left[-
{\dfrac {4897161727}{10000000000}},-{\dfrac {4822526949}{10000000000}}\right]
\\\noalign{\medskip}\left[{\dfrac {616963279}{100000000000}},{\dfrac {
1002395789}{100000000000}}\right]\end {array} \right)
$
\\[10pt]
$
W_{{5}}= \left( \begin {array}{c} \left[{\dfrac {323160259}{1250000000}},{
\dfrac {2727361243}{10000000000}}\right]\\\noalign{\medskip}\left[-{\dfrac {
1486861643}{2500000000}},-{\dfrac {5875949419}{10000000000}}\right]
\\\noalign{\medskip}\left[{\dfrac {116759761}{5000000000}},{\dfrac {109700733
}{4000000000}}\right]\\\noalign{\medskip}\left[-{\dfrac {585736263}{5000000000}},-
{\dfrac {2143830967}{25000000000}}\right]\\\noalign{\medskip}\left[-{\dfrac {
36387877}{31250000}},-{\dfrac {225679327}{200000000}}\right]
\\\noalign{\medskip}\left[-{\dfrac {1002395789}{100000000000}},-{\dfrac {
616963279}{100000000000}}\right]\\\noalign{\medskip}\left[{\dfrac {4813603601}{
10000000000}},{\dfrac {1215175283}{2500000000}}\right]\end {array} \right),
W_{{6}}= \left( \begin {array}{c} \left[-{\dfrac {6654530721}{10000000000}},
-{\dfrac {6522387553}{10000000000}}\right]\\\noalign{\medskip}\left[-{\dfrac {
428648367}{2500000000}},-{\dfrac {792095233}{5000000000}}\right]
\\\noalign{\medskip}\left[{\dfrac {9385595361}{100000000000}},{\dfrac {
9959577971}{100000000000}}\right]\\\noalign{\medskip}\left[-{\dfrac {4897161727}{
10000000000}},-{\dfrac {4822526949}{10000000000}}\right]\\\noalign{\medskip}\left[
-{\dfrac {1002395789}{100000000000}},-{\dfrac {616963279}{100000000000}}
\right]\\\noalign{\medskip}\left[-{\dfrac {8828398113}{10000000000}},-{\dfrac {
4195755701}{5000000000}}\right]\\\noalign{\medskip}\left[{\dfrac {171506479}{
2000000000}},{\dfrac {1171472521}{10000000000}}\right]\end {array} \right)
$
\\[10pt]
$
W_{{7}}= \left( \begin {array}{c} \left[-{\dfrac {545472247}{2000000000}},-{
\dfrac {16158013}{62500000}}\right]\\\noalign{\medskip}\left[{\dfrac {2937974709}{
5000000000}},{\dfrac {5947446571}{10000000000}}\right]\\\noalign{\medskip}\left[{
\dfrac {2335195223}{100000000000}},{\dfrac {2742518293}{100000000000}}\right]
\\\noalign{\medskip}\left[{\dfrac {616963279}{100000000000}},{\dfrac {
1002395789}{100000000000}}\right]\\\noalign{\medskip}\left[{\dfrac {4813603601}{
10000000000}},{\dfrac {1215175283}{2500000000}}\right]\\\noalign{\medskip}\left[{
\dfrac {171506479}{2000000000}},{\dfrac {1171472521}{10000000000}}\right]
\\\noalign{\medskip}\left[-{\dfrac {582206031}{500000000}},-{\dfrac {
564198319}{500000000}}\right]\end {array} \right)
$
\\[10pt]
}

Through Theorem \ref{th:nonpos},
we can judge that $\mathcal{W}$ is nonnegative semidefinite,
and when the disturbance enlarges very little,
it is still true.
So by Theorem \ref{th:notextremepoint}, we know that
values of $f$ cannot attain the maximum in $\mathcal{D}_{p}$, i.e.,

\begin{proposition} \label{prop:pyramid}
Maximum of function $f$ can not attain in domain $\mathcal{D}_{p}$, i.e.,
\begin{equation} \label{eqn:pyramid}
f(\phi_1, \phi_2, \theta_2, \phi_3, \theta_3, \phi_4, \theta_4) < fmax, \quad
(\phi_1, \phi_2, \theta_2, \phi_3, \theta_3, \phi_4, \theta_4) \in \mathcal{D}_{p}.
\end{equation}
\end{proposition}

\section{Other domains}
Now, we are to prove the following strict inequality,
\begin{equation} \label{strictinequality}
\begin{split}
& f(\phi_1, \phi_2, \theta_2, \phi_3, \theta_3, \phi_4, \theta_4) < fmax,\\
& \mbox{where } (\phi_1, \phi_2, \theta_2, \phi_3, \theta_3, \phi_4, \theta_4) \in
(\mathcal{D}^{(1)} \cup \mathcal{D}^{(2)}) \backslash (\mathcal{D}_{bp} \cup \mathcal{D}_{p}).
\end{split}
\end{equation}

Algorithms in this section are implemented by procedures in the Maple package
\textbf{fivepoints}, for the code, see appendix.

\subsection{Branch and bound strategies} \label{5points:bb}
We check domains over which variables take using the interval method,
more precisely, we compute the interval value of the interval mapping
corresponding some functions through interval arithmetic,
properties of this interval may suggest that,
when variables take values in this domain,
function $f$ has no stationary point,
or its maximum is less than the value corresponding the
bipyramid configuration, or it is not necessary to consider the
case for symmetries. All in all,
function values in this domain
cannot be greater than $fmax$
(these verification methods are implemented
by procedure \textbf{ischecked} in the package \textbf{fivepoints}).
The followings are methods we used to exclude domains contained in
$\mathcal{D}^{(1)}$ and $\mathcal{D}^{(2)}$.

\begin{enumerate}
\item \label{pointc}
(by Condition \ref{as:pointc})
Verify that $C$ is below $E$.

\item \label{halfsphere}
(by Condition \ref{as:halfsphere})
Verify that 5 points are in the same half sphere.

\item \label{secondlength}
(by Condition \ref{as:secondlength})
Verify that $AB$ is not the second largest distance.

\item \label{mindis}
(by Condition \ref{as:mindis})
Verify that the distance between some two points is less than $\dfrac{2}{15}$.

\item \label{totaldis}
Verify that the upper bound of function values are less than
$fmax$ (see Eq. \eqref{fmax}).

\item \label{derivative}
(by Proposition \ref{prop:parder})
Verify that some partial derivative of $f$ does not change signs
in this domain.

\item \label{totaldefver}
Compute the interval Hessian matrix corresponding this domain,
and determine its negative definiteness through Theorem \ref{th:pstv}.

\item \label{totaldefeig}
Compute the interval Hessian matrix corresponding this domain,
and determine its negative definiteness through Theorem \ref{th:posdefeig}.

\item \label{totalnotdef}
Compute the interval Hessian matrix corresponding this domain,
and determine its nonnegative definiteness through Theorem \ref{th:nonpos}.

\item \label{gradient}
Determine there exists no maximal point
in this domain through Theorem \ref{th:gradient}.

\end{enumerate}

Different methods should be used in different domains,
for example, methods \ref{totaldefver} and \ref{totaldefeig}
should be used first near points corresponding
bipyramid distribution, then others
(\ref{gradient}, \ref{totaldis}, \ref{derivative}) can be used;
while method \ref{totalnotdef} should be used first
near points corresponding pyramid distribution, then
others (\ref{gradient}, \ref{totaldis}, \ref{derivative})
can be used; and for generic domains,
methods can be used in turns as
\ref{pointc}, \ref{halfsphere}, \ref{mindis},
\ref{secondlength}, \ref{totaldis}, \ref{gradient}, \ref{derivative}.

For a domain to be verified,
we choose appropriate verification methods and the verification order,
if verifications are not successful, we subdivide the interval whose
width is maximal into two equal intervals, and
verify the two subdomains recursively.
We set a positive number, if
the largest interval width of a domain
we get in the above process is less than this number,
we stop subdividing this domain, and record it,
this domain may contain distributions of points corresponding larger
distance sums then the maximal distance sum conjectured.
This process terminates when all domains have been
verified.
If all domains are verified successful, and
no domain is contained in the record list,
then we have proved the conjecture in fact.
The complete algorithm is described below
(implemented as the procedure \textbf{spchecked}
in the package \textbf{fivepoints}):

\begin{algorithm}[H]
\scriptsize
\dontprintsemicolon
\linesnumbered
\caption{CheckDomain}
\KwIn{intervallist, methods, notcheckbipyramid, notcheckpyramid}
\KwOut{true/false}

\Begin{
checkbipyramid := not notcheckbipyramid; checkpyramid := not notcheckpyramid;
checkmethods := methods

\If{checkbipyramid}{
\If(
\;\#\textit{bipyramidintervallist is the rectangular
domain we excluded first near the bipyramid distribution.}
)
{intervallist is contained in domain bipyramidintervallist}{
add $[-1]$ to checkprocess\;
\#\textit{record the checking process}\;
\Return{true}
}
\If{some variable interval in intervallist
disjoints the corresponding variable interval
in bipyramidintervallist}
{checkbipyramid := false}
}

\If{checkpyramid}{
\If(\;\#\textit{pyramidintervallist is the rectangular
domain we excluded first near the pyramid distribution.}
)
{intervallist is contained in domain pyramidintervallist}{
add $[0]$ to checkprocess; \Return{true}
}
\If{some variable interval in intervallist
disjoints the corresponding variable interval
in pyramidintervallist}
{checkpyramid := false}
}

\While{$methods_i \in methods$}{
\If{checking domain intervallist successfully
through the method $methods_i$}{
add $[i]$ to checkprocess; \Return{true}
}
}

dim := the widest interval position in interval vector intervallist\;
\If{the width of the dim-th interval in
intervallist $<$ $\frac{1}{1000}$}{
add intervallist to notchecked\;
\#\textit{notchecked records domains cannot be checked successfully}\;
add $[-4]$ to checkprocess; \Return{false}
}
add dim to checkprocess\;
\#\textit{dim designates the variable to be subdivided, and is to be
recorded in checking process}\;
subintervallist := subdivide the domain intervallist in the dim-th interval\;
cur := true\;
\ForAll{$subintervallist_i \in subintervallist$}{
\If{not CheckDomain($subintervallist_i$, checkmethods, not checkbipyramid,
not checkpyramid)}{cur := false}
}
\If{cur}{add $[-2]$ to checkprocess}
\Else{add $[-3]$ to checkprocess}

\Return{cur}
}
\end{algorithm}

\subsection{Verification process}
In order to subdivide domains into appropriate widths,
we act some experiments first, finally we subdivide domains as follows:
for $\mathcal{D}^{(1)}$ and $\mathcal{D}^{(2)}$
we divide in \S\,\ref{5points:math},
each domain is subdivided the way that each interval of it is trisected,
so we get $3^7=2817$ subdomains each, denote them respectively by:
$$\mathcal{D}^{(1)}_1,\mathcal{D}^{(1)}_2,\ldots,\mathcal{D}^{(1)}_{2187},$$
and
$$\mathcal{D}^{(2)}_1,\mathcal{D}^{(2)}_2,\ldots,\mathcal{D}^{(2)}_{2187}.$$
If some of these subdomains are
difficult to verify successfully,
we can again subdivide them the same way.
Actually, the following domains need to subdivide again:
\begin{equation} \label{eqn:subdivision}
\begin{split}
&\mathcal{D}^{(1)}_{62}, \mathcal{D}^{(1)}_{158}, \mathcal{D}^{(1)}_{239},
\mathcal{D}^{(1)}_{863}, \mathcal{D}^{(1)}_{1102}, \mathcal{D}^{(1)}_{1105},
\mathcal{D}^{(1)}_{1106}, \mathcal{D}^{(1)}_{2114}, \mathcal{D}^{(1)}_{2132},\\
&\mathcal{D}^{(1)}_{1105-1101},
\mathcal{D}^{(1)}_{1106-834}, \mathcal{D}^{(1)}_{1106-861},
\mathcal{D}^{(1)}_{1106-1099}, \mathcal{D}^{(1)}_{1106-1100},\\
&\mathcal{D}^{(1)}_{1105-1101-1100},
\mathcal{D}^{(1)}_{1106-834-725}, \mathcal{D}^{(1)}_{1106-834-726},\\
&\mathcal{D}^{(1)}_{1106-834-725-1752},
\mathcal{D}^{(1)}_{1106-834-726-1507}, \mathcal{D}^{(1)}_{1106-834-726-1750},
\end{split}
\end{equation}
where $\mathcal{D}^{(1)}_{1105-1101}$ denotes the 1101-th subdomain in all 2187
subdomains of $\mathcal{D}^{(1)}_{1105}$, other similar notations are understood the same way.

\subsection{Algorithm implementations}
The Maple Package \textbf{fivepoints} implements algorithms described in above sections.
For the detailed code, see Appendix \ref{fivepoints}.

\section{Conclusion}
The following is verification time for various domains
(may differ on different computers, it is the time used by computers
with Pentium IV 3.0 GHz CPU, and 1 GB RAM):
\begin{enumerate}
\item
Time used to verify domain $\mathcal{D}^{(1)}$:
$782534.203$ seconds.

\item
Time used to verify domain $\mathcal{D}^{(2)}$:
$8797.600$ seconds.

\item
Total time:
$791331.803$ seconds.
\end{enumerate}

This completes the proof of the problem of spherical distribution of 5 points.

\clearpage

\begin{appendix}
\section{Maple code}
\subsection{Module IntervalArithmetic} \label{intervalarithmetic}
\begin{maplettyout}
# IntervalArithmetic: a Maple package used for interval computations.
# $revision: 1.0.3.4$
IntervalArithmetic := module()
  export ulp, fulb, rfulb0, rfulb,
    `Evalr/add`, `Evalr/multiply`, `Evalr/power`,
    `Evalr/sin`, `Evalr/cos`, `Evalr/tan`, `Evalr/cot`,
    `Evalr/arcsin`, `Evalr/arccos`, `Evalr/arctan`, `Evalr/arccot`,
    `Evalr/exp`, `Evalr/ln`, `Evalr/powexp`,
    `Evalr/shake`, Evalr,
    isdef, vert, inthull, intwidth, maxwidthdim, intsbdv, sortpos, contain:
  local init:
  option package, load = init:

##############################################################################
# initialization
##############################################################################

# the initialization procedure
init := proc()
global `type/interval`, `type/int_ext`, `type/rat_ext`,
  `type/cons`, `type/consnum`, `type/ratpar`,
  `type/exp_user`, `type/ln`,
  `convert/ft2rat`, truncatenegativepart, p1, p2:

# defining the datatype of interval
`type/interval` := proc( inv )
  if type( inv, list( rat_ext ) ) and nops( inv ) = 2
  and inv[ 1 ] <= inv[ 2 ] then
    true
  elif inv = [ ] then
    true
  else
    false
  fi
end:

# extended datatype of integer numbers
`type/int_ext` := proc( x )
  evalb(  type( x, integer ) or x = -infinity or x = infinity  )
end:

# extended datatype of rational numbers
`type/rat_ext` := proc( x )
  evalb(  type( x, rational ) or x = -infinity or x = infinity  )
end:

# datatype that containing gamma, Catalan or Pi
`type/cons` := proc( x )
  member( x, { gamma, Catalan, Pi } )
end:

# datatype of constant numbers
`type/consnum` := proc( x )
  local c, r:
  global constants:
  c := constants:
  constants := gamma, Catalan, Pi, infinity:
  r := type( x, constant ):
  constants := c:
  r
end:

# datatype of functions whose parameters are rational numbers
`type/ratpar` := proc( f, tp::Or( set( type ), type ) )
  type( f, tp ) and type( [ op( f ) ], list( rational ) )
end:

# datatype of natural exponent
`type/exp_user` := proc( f )
  evalb( op( 0, f ) = `exp` )
end:

# datatype of natural logarithm
`type/ln` := proc( f )
  evalb( op( 0, f ) = `ln` )
end:

# converting from float point numbers(ranges) to rational numbers(intervals)
`convert/ft2rat` := proc( inv )
  local t:
  if type( inv, float ) then
    convert( inv, rational, exact )
  elif type( inv, rat_ext ) then
    inv
  elif type( inv, list( Or( float, rat_ext ) ) ) and nops( inv ) = 2 then
    t := map( procname, inv ):
    if t[ 1 ] <= t[ 2 ] then t
    else [ ]
    fi
  elif type( inv, list ) or type( inv, Matrix ) then
    map( procname, inv )
  else
    error "invalid argument: 
  fi
end:

# If truncatenegativepart is set to ture,
# then in interval power arithmetic,
# negative parts of base intervals are truncated.
# If truncatenegativepart is set to false,
# then an error is generated when encountering
# a interval power arithmetic whose base contains a negative part.
truncatenegativepart := false:

# rational lower bound of Pi
p1 := rfulb0( Pi, 'r', 'l' ):

# rational upper bound of Pi
p2 := rfulb0( Pi, 'r', 'u' ):

end:

# the error of a float point number
ulp := proc( x )
  if type( x, float ) then
    if x = Float( infinity ) or x = -Float( infinity ) then 0
    else Float( 1, length( op( 1, x ) )+op( 2, x )-Digits )
    fi
  elif type( x, list ) or type( x, Matrix ) then map( procname, x )
  else error "invalid argument: 
  fi
end:

# the float upper or lower bounds of a float number
fulb := proc( x, ul::{identical( u ), identical( l )} )
  local t:
  if type( x, float ) then
    if x = Float( infinity ) or x = -Float( infinity ) then
      x
    else
      # not using "eval" may generate an error,
      # may be a bug in maple( evaluation rule in Float ).
      if type( ul, identical( u ) ) then
        eval( x+Float( 1, length( op( 1, x ) )+op( 2, x )-Digits ) )
      else
        eval( x-Float( 1, length( op( 1, x ) )+op( 2, x )-Digits ) )
      fi
    fi
  elif type( x, int_ext ) then
    x
  elif type( x, consnum ) then
    procname( evalf( x ), args[ 2..-1 ] )
  elif type( x, list ) or type( x, Matrix ) then
    map( procname, x, args[ 2..-1 ] )
  else
    error "invalid argument: 
  fi
end:

# the rational(or float) upper(or lower) bound of
# a number that can be computed to any precision in Maple system
rfulb0 := proc( x, rf::{identical( r ), identical( f )},
ul::{identical( u ), identical( l )} )
  local p, q, n, t:
  if type( x, int_ext ) then
    x
  elif type( x, rational ) then
    if rf = 'r' then x
    else fulb( x, ul )
    fi
  elif type( x, Or( float, cons,
  ratpar( { trig, arctrig, exp_user, ln } ) ) )
  or ( op( 0, x ) in { `exp`, `ln`, `arctan`, `arccot` }
  and op( x ) = infinity )
  or ( op( 0, x ) in { `exp`, `arctan`, `arccot` }
  and op( x ) = -infinity ) then
    if rf = 'f' then fulb( x, ul )
    else convert( fulb( x, ul ), ft2rat )
    fi
  elif type( x, ratpar( `^` ) ) then
    q := op( 1, x ): n := op( 2, x ):
    if q < 0 then    # e.g. ( -16 )^( 1/3 )
      if type( n, integer ) then
        procname( q^n, args[ 2..-1 ] )
      elif type( numer( n ), even ) then
        procname( (-q)^n, args[ 2..-1 ] )
      elif type( numer( n ), odd ) and type( denom( n ), odd ) then
        procname( -( -q )^n, args[ 2..-1 ] )
      else
        error "negative radicand"
      fi
    else
      if rf = 'f' then fulb( x, ul )
      else convert( fulb( x, ul ), ft2rat )
      fi
    fi
  elif type( x, `+` ) and
  type( [ op( x ) ], list(
  Or(  rational, cons, ratpar( {trig, arctrig, `^`} ) )  ) ) then
    `+`(  op(  map(  procname, [ op( x ) ], args[ 2..-1 ] ) ) )
  elif type( x, `*` ) and type( [ op( x ) ],
  list(  Or(  rational, ratpar( `^` ) )  ) ) then
    if hastype( remove( type, [ op( x ) ], rational ), negative ) then
      procname( map( convert, x, surd ), args[ 2..-1 ] )
    elif op( select( type, [ op( x ) ], rational ) ) < 0 then
      procname(  map(  procname, x, 'r',
      `if`( ul = 'u', 'l', 'u' ) ), args[ 2..-1 ] )
    else
      procname(  map(  procname, x, 'r',
      `if`( ul = 'u', 'u', 'l' ) ), args[ 2..-1 ] )
    fi
  elif type( x, list ) or type( x, Matrix ) then
    map( procname, x, args[ 2..-1 ] )
  else
    error "cannot determine the 
    `if`( rf = 'r', rational, float ), `if`( ul = 'u', upper, lower ), x
  fi
end:

# the rational(or float) upper(or lower) bound of a number
rfulb := proc( x, rf::{identical( r ), identical( f )},
ul::{identical( u ), identical( l )} )
  local i, t:
  if type( x, consnum ) then
    t := `Evalr/shake`( x ):
    if type( t, 'interval' ) then t := op( `if`( ul = 'l', 1, 2 ), t ) fi:
    if rf = 'f' then t := rfulb0( t, 'f', ul ) fi:
    t
  elif type( x, list ) or type( x, Matrix ) then
    map( procname, x )
  else error "invalid argument: 
  fi:
end:

##############################################################################
# Following procedures define interval arithmetic
##############################################################################

# interval addition
`Evalr/add` := proc( L::list )
  local n, i, lb, ub:
  if member( [ ], L) then return [ ] fi:
  n := nops( L ):
  lb := 0: ub := 0:
  for i to n do
    if type( L[ i ], 'interval' ) then
      lb := lb + L[ i, 1 ]: ub := ub + L[ i, 2 ]:
    elif type( L[ i ], rat_ext ) then
      lb := lb + L[ i ]: ub := ub + L[ i ]:
    else
      error "unrecognized argument: 
    fi
  od:

  # Large integer may occur after computation with rational numbers,
  # we adjust intervals outside, expressing denominators of interval
  # vertexes with integers whose digits is not greater then Digits.
  lb := convert( rfulb0( lb, 'f', l ), ft2rat ):
  ub := convert( rfulb0( ub, 'f', u ), ft2rat ):
  return [ lb, ub ]
end:

# arithmetic multiplication
`Evalr/multiply` := proc( L::list )
  local n, i, lb, ub, a, b, tmp:
  if member( [ ], L) then return [ ] fi:
  n := nops( L ):
  lb := 1: ub := 1:
  for i to n do
    if type( L[ i ], 'interval' ) then
      a := L[ i, 1 ]: b := L[ i, 2 ]:
      if a >= 0 then
        if lb >= 0 then lb := lb*a: ub := ub*b:
        elif ub <= 0 then lb := lb*b: ub := ub*a:
        else lb := lb*b: ub := ub*b:
        fi:
      elif b <= 0 then
        if lb >= 0 then tmp := lb: lb := ub*a: ub := tmp*b:
        elif ub <= 0 then tmp := lb: lb := ub*b: ub := tmp*a:
        else tmp := lb: lb := ub*a: ub := tmp*a:
        fi:
      else    # a < 0 < b
        if lb >= 0 then lb := ub*a: ub := ub*b:
        elif ub <= 0 then ub := lb*a: lb := lb*b:
        else  # lb < 0 < ub
          tmp := lb:
          if ( lb*b ) <= ( ub*a )    # lb*b and ub*a will not be 0*infinity
            then lb := lb*b:
          else lb := ub*a:
          fi:
          if ( tmp*a ) <= ( ub*b )
            then ub := ub*b:
          else ub := tmp*a:
          fi:
        fi:
      fi:
    elif type( L[ i ], rational ) then
      if ( L[ i ] ) >= 0 then lb := lb*L[ i ]: ub := ub*L[ i ]:
      else tmp := lb: lb := ub*L[ i ]: ub := tmp*L[ i ]:
      fi:
    else error "unrecognized argument: 
    fi:
  od:
  lb := convert( rfulb0( lb, 'f', 'l' ), ft2rat ):
  ub := convert( rfulb0( ub, 'f', 'u' ), ft2rat ):
  return [ lb, ub ]
end:

# interval power
`Evalr/power` := proc( a, n::rational )
  local lb, ub, tmp:
  if n = 0 then return 0 fi:
  if a = [ ] then return [ ] fi:
  if type( a, 'interval' ) then
    lb := a[ 1 ]: ub := a[ 2 ]:
    if n > 0 then
      if lb >= 0 then
        lb := lb^n: ub := ub^n:
      elif ub < 0 then
        if type( numer( n ), odd ) then
          if type( denom( n ), odd ) then
            lb := lb^n: ub := ub^n:
          else    # e.g. sqrt( [ -2, -1 ] )
            error "negative radicand"
          fi:
        else tmp := lb: lb := ub^n: ub := tmp^n:
        fi:
      else    # lb < 0 =< ub
        if type( numer( n ), odd ) then
          if type( denom( n ), odd ) then
            lb := lb^n: ub := ub^n:
          else   # e.g. sqrt( [ -2, 1 ] )
            if truncatenegativepart then
              lb := 0: ub := ub^n:
            else
              error "negative radicand"
            fi
          fi:
        else    # e.g. ( [ -2, 1 ] )^2
          if ( -lb ) <= ub then ub := ub^n:
          else ub := ( -lb )^n:
          fi:
          lb := 0:
        fi:
      fi:
    else    # n < 0
      if lb = 0 then
        if ub = 0 then
          lb := -infinity: ub := infinity:
        else
          lb := ub^n: ub := infinity:
        fi:
      elif lb > 0 then tmp := lb: lb := ub^n: ub := tmp^n:
      elif ub <= 0 then
        if type( numer( n ), odd ) then
          if type( denom( n ), odd ) then
            if ub=0 then ub := lb^n: lb := -infinity:
            else tmp := lb: lb := ub^n: ub := tmp^n:
            fi:
          else    # e.g. 1/sqrt( [ -2,-1 ] )
            error "negative radicand"
          fi:
        else
          if ub = 0 then lb := lb^n: ub := infinity:
          else lb := lb^n: ub := ub^n:
          fi:
        fi:
      else    # lb < 0 < ub
        if type( numer( n ), odd ) then    # e.g. 1/sqrt( [ -2, 1 ] )
          if truncatenegativepart then
            lb := ub^n: ub := infinity:
          else
            error "negative radicand"
          fi
        else    # e.g. 1/( [ -2, 1 ] )^2.
          if ( -lb ) <= ub then lb := ub^n:
          else lb := ( -lb )^n:
          fi:
          ub := infinity:
        fi:
      fi:
    fi:
  elif type( a, rational ) then
    return [ rfulb0( a^n, 'r', 'l' ), rfulb0( a^n, 'r', 'u' ) ]
  else error "invalid argument: 
  fi:
  if type( lb, Not( rational ) ) then lb := rfulb0( lb, 'r', 'l' )
  else lb := convert( rfulb0( lb, 'f', 'l' ), ft2rat )
  fi:
  if type( ub, Not( rational ) ) then ub := rfulb0( ub, 'r', 'u' )
  else ub := convert( rfulb0( ub, 'f', 'u' ), ft2rat )
  fi:
  return [ lb, ub ]
end:

# interval sine
`Evalr/sin` := proc( a )
  local lb, ub, tmp:
  global p1, p2:
  if a = [ ] then return [ ] fi:
  if type( a, 'interval' ) then
    lb := a[ 1 ]: ub := a[ 2 ]:
    if ub <= 0 then
      tmp := procname( [ -ub, -lb ] ):
      return [ -tmp[ 2 ], -tmp[ 1 ] ]
    fi:
    lb := lb/`if`( lb >0, p2, p1 ): ub := ub/`if`( ub >0, p1, p2 ):
    tmp := floor( lb/2 ): lb := lb-2*tmp: ub := ub-2*tmp:
    if ub >= ( lb+2 ) then lb := -1: ub := 1:
    else
      if ( lb-1/2 ) < 0 then
        if ( ub-1/2 ) < 0 then
          # 1/2*Pi > lb*Pi > lb*p1 > 0
          # ( ub-1/2 )*Pi < ( ub-1/2 )*p1 < 0 ==>
          # 0 < ub*Pi < 1/2*Pi+( ub-1/2 )*p1 < Pi/2
          lb := rfulb0( sin( lb*p1 ), 'r', 'l' ):
          ub := rfulb0( cos( ( 1/2-ub )*p1 ), 'r', 'u' ):
        elif ( ub-1+lb ) < 0 then
          # 1/2*Pi > lb*Pi > lb*p1 > 0
          lb := rfulb0( sin( lb*p1 ), 'r', 'l' ): ub := 1:
        elif ( ub-1 ) < 0 then
          # ( ub-1 )*Pi < ( ub-1 )*p1 < 0 ==>
          # Pi/2 < ub*Pi < Pi+( ub-1 )*p1 < Pi
          lb := rfulb0( sin( ( 1-ub )*p1 ), 'r', 'l' ): ub := 1:
        elif ( ub-3/2 ) < 0 then
          # ( ub-3/2 )*Pi < ( ub-3/2 )*p1 < 0 ==>
          # Pi <= ub*Pi < 3/2*Pi+( ub-3/2 )*p1 < 3/2*Pi
          lb := -rfulb0( cos( ( 3/2-ub )*p1 ), 'r', 'u' ): ub := 1:
        else
          lb := -1: ub := 1:
        fi:
      elif ( lb-1 ) < 0 then
        if ( ub-3/2 ) < 0 then
          # ( ub-3/2 )*Pi < ( ub-3/2 )*p1 <0 ==>
          # Pi/2 <= ub*Pi < 3/2*Pi+( ub-3/2 )*p1 < 3/2*Pi
          # ( lb-1/2 )*Pi >= ( lb-1/2 )*p1 >=0 ==>
          # Pi > lb*Pi >= Pi/2+( lb-1/2 )*p1 >= Pi/2
          tmp := lb:
          lb := -rfulb0( cos( ( 3/2-ub )*p1 ), 'r', 'u' ):
          ub := rfulb0( cos( ( tmp-1/2 )*p1 ), 'r', 'u' ):
        elif ( ub-3+lb ) < 0 then
          # ( lb-1/2 )*Pi >= ( lb-1/2 )*p1 >=0 ==>
          # Pi > lb*Pi >= Pi/2+( lb-1/2 )*p1 >= Pi/2
          ub := rfulb0( cos( ( lb-1/2 )*p1 ), 'r', 'u' ): lb := -1:
        elif ( ub-5/2 ) < 0 then
          # ( ub-5/2 )*Pi < ( ub-5/2 )*p1 < 0 ==>
          # 2*Pi <= ub*Pi < 5/2*Pi+( ub-5/2 )*p1 < 5/2*Pi
          lb := -1: ub := rfulb0( cos( ( 5/2-ub )*p1 ), 'r', 'u' ):
        else lb := -1: ub := 1:
        fi:
      elif ( lb-3/2 ) < 0 then
        if ( ub-3/2 ) < 0 then
          # ( ub-3/2 )*Pi < ( ub-3/2 )*p1 < 0 ==>
          # Pi < ub*Pi < 3/2*Pi+( ub-3/2 )*p1 < 3/2*Pi
          # ( lb-1 )*Pi >= ( lb-1 )*p1 >= 0 ==>
          # 3/2*Pi > lb*Pi >= Pi+( lb-1 )*p1 >= Pi
          tmp := lb:
          lb := -rfulb0( cos( ( 3/2-ub )*p1 ), 'r', 'u' ):
          ub := -rfulb0( sin( ( tmp-1 )*p1 ), 'r', 'l' ):
        elif ( ub-3+lb ) < 0 then
          # ( lb-1 )*Pi >= ( lb-1 )*p1 >=0 ==>
          # 3/2*Pi > lb*Pi >= Pi+( lb-1 )*p1 >= Pi
          ub := -rfulb0( sin( ( lb-1 )*p1 ), 'r', 'l' ): lb := -1:
        elif ( ub-5/2 ) < 0 then
          # ( ub-5/2 )*Pi < ( ub-5/2 )*p1 < 0 ==>
          # 3/2*Pi <= ub*Pi < 5/2*Pi+( ub-5/2 )*p1 < 5/2*Pi
          lb := -1: ub := rfulb0( cos( ( 5/2-ub )*p1 ), 'r', 'u' ):
        else lb := -1: ub := 1:
        fi:
      elif ( ub-5/2 ) < 0 then
        # ( lb-3/2 )*Pi >= ( lb-3/2 )*p1 >=0 ==>
        # 2*Pi > lb*Pi >= 3/2*Pi+( lb-3/2 )*p1 >= 3/2*Pi
        # ( ub-5/2 )*Pi < ( ub-5/2 )*p1 < 0 ==>
        # 3/2*Pi <= ub*Pi < 5/2*Pi+( ub-5/2 )*p1 < 5/2*Pi
        lb := -rfulb0( cos( ( lb-3/2 )*p1 ), 'r', 'u' ):
        ub := rfulb0( cos( ( 5/2-ub )*p1 ), 'r', 'u' ):
      elif ( ub-5+lb ) < 0 then
        # ( lb-3/2 )*Pi >= ( lb-3/2 )*p1 >=0 ==>
        # 2*Pi > lb*Pi >= 3/2*Pi+( lb-3/2 )*p1 >= 3/2*Pi
        lb := -rfulb0( cos( ( lb-3/2 )*p1 ), 'r', 'u' ): ub := 1:
      elif ( ub-7/2 ) < 0 then
        # ( ub-7/2 )*Pi < ( ub-7/2 )*p1 < 0 ==>
        # 3*Pi <= ub*Pi < 7/2*Pi+( ub-7/2 )*p1 < 7/2*Pi
        lb := -rfulb0( cos( ( 7/2-ub )*p1 ), 'r', 'u' ): ub := 1:
      else lb := -1: ub := 1:
      fi:
    fi:
  elif type( a, rational ) then
    return [ rfulb0( 'sin'( a ), 'r', 'l' ), rfulb0( 'sin'( a ), 'r', 'u' ) ]
  else error "invalid argument: 
  fi:
  return [ lb, ub ]
end:

# interval cosine
`Evalr/cos` := proc( a )
  global p1, p2:
  if a = [ ] then return [ ] fi:
  if type( a, 'interval' ) then
    `Evalr/sin`( [ -a[ 2 ] + p1/2, -a[ 1 ] + p2/2 ] )
  elif type( a, rational ) then
    # using 'cos' to suppress automatic simplification
    [ rfulb0( 'cos'( a ), 'r', 'l' ), rfulb0( 'cos'( a ), 'r', 'u' ) ]
  else
    error "invalid argument: 
  fi:
end:

# interval tangent
`Evalr/tan` := proc( a )
  local lb, ub, tmp:
  global p1, p2:
  if a = [ ] then return [ ] fi:
  if type( a, 'interval' ) then
    lb := a[ 1 ]: ub := a[ 2 ]:
    lb := lb/`if`( lb > 0, p2, p1 ): ub := ub/`if`( ub > 0, p1, p2 ):
    tmp := floor( lb ): lb := lb - tmp: ub := ub - tmp:
    if ub >= ( lb+1 ) then lb := -infinity: ub := infinity:
    else
      if ( lb-1/2 ) < 0 then
        if ( ub-1/2 ) < 0 then
          # 0<lb*p1<lb*Pi<a[1]-tmp*Pi<a[2]-tmp*Pi<ub*Pi<ub*p2
          lb := rfulb0( 'tan'( a[ 1 ] ), 'r', 'l' ):
          ub := rfulb0( 'tan'( a[ 2 ] ), 'r', 'u' ):
        else
          lb := -infinity: ub := infinity:
        fi:
      elif ( ub-3/2 ) < 0 then
        # 0<lb*p1<lb*Pi<a[1]-tmp*Pi<a[2]-tmp*Pi<ub*Pi<ub*p2
        lb := rfulb0( 'tan'( a[ 1 ] ), 'r', 'l' ):
        ub := rfulb0( 'tan'( a[ 2 ] ), 'r', 'u' ):
      else
        lb := -infinity: ub := infinity:
      fi:
    fi:
  elif type( a, rational ) then
    return [ rfulb0( 'tan'( a ), 'r', 'l' ), rfulb0( 'tan'( a ), 'r', 'u' ) ]
  else
    error "invalid argument: 
  fi:
  return [ lb, ub ]
end:

# interval cotangent
`Evalr/cot` := proc( a )
  global p1, p2:
  if a = [ ] then return [ ] fi:
  if type( a, 'interval' ) then
    `Evalr/tan`( [ -a[ 2 ] + p1/2, -a[ 1 ] + p2/2 ] )
  elif type( a, rational ) then
    [ rfulb0( 'cot'( a ), 'r', 'l' ), rfulb0( 'cot'( a ), 'r', 'u' ) ]
  else
    error "invalid argument: 
  fi:
end:

# interval arc sine
`Evalr/arcsin` := proc( inv )
  local lb, ub:
  global p1, p2:
  if inv = [ ] then return [ ] fi:
  if type( inv, rational ) then
    if inv >= -1 and inv <= 1 then
      [  rfulb0( 'arcsin'( inv ), 'r', 'l' ),
      rfulb0( 'arcsin'( inv ), 'r', 'u' ) ]
    else
      error "invalid argument: 
    fi
  elif type( inv, 'interval' ) then
    lb := inv[ 1 ]: ub := inv[ 2 ]:
    if not ( type( lb, rational ) and type( ub, rational ) ) then
      error "invalid argument: 
    fi:
    if lb > 1 or ub < -1 then
      error "invalid argument: 
    elif lb < -1 and ub > 1 then
      lb := -p2:
      ub := p2:
    elif lb < -1 then
      lb := -p2: ub := rfulb0( 'arcsin'( ub ), 'r', 'u' ):
    elif ub > 1 then
      lb := rfulb0( 'arcsin'( lb ), 'r', 'l' ): ub := p2:
    else
      lb := rfulb0( 'arcsin'( lb ), 'r', 'l' ):
      ub := rfulb0( 'arcsin'( ub ), 'r', 'u' ):
    fi:
    [ lb, ub ]
  else
    error "invalid argument: 
  fi
end:

# interval arc cosine
`Evalr/arccos` := proc( inv )
  local tmp:
  global p1, p2:
  if inv = [ ] then return [ ] fi:
  if type( inv, rational ) then
    if inv >= -1 and inv <= 1 then
      [ rfulb0( 'arccos'( inv ), 'r', 'l' ),
      rfulb0( 'arccos'( inv ), 'r', 'u' ) ]
    else
      error "invalid argument: 
    fi:
  elif type( inv, 'interval' ) then
    tmp := `Evalr/arcsin`( inv ):
    [ p1/2-tmp[ 2 ], p2/2-tmp[ 1 ] ]
  else
    error "invalid argument: 
  fi:
end:

# interval arc tangent
`Evalr/arctan` := proc( inv )
  if inv = [ ] then return [ ] fi:
  if type( inv,rat_ext ) then
    [ rfulb0( 'arctan'( inv ), 'r', 'l' ),
    rfulb0( 'arctan'( inv ), 'r', 'u' ) ]
  elif type( inv, 'interval' ) then
    [ rfulb0( 'arctan'( inv[ 1 ] ), 'r', 'l' ),
    rfulb0( 'arctan'( inv[ 2 ] ), 'r', 'u' ) ]
  else
    error "invalid argument: 
  fi
end:

# interval arc cotangent
`Evalr/arccot` := proc( inv )
  local tmp:
  global p1, p2:
  if inv = [ ] then return [ ] fi:
  if type( inv, rat_ext ) then
    [ rfulb0( 'arccot'( inv ), 'r', 'l' ),
    rfulb0( 'arccot'( inv ), 'r', 'u' ) ]
  elif type( inv, 'interval' ) then
    tmp := `Evalr/arctan`( inv ):
    [ p1/2-tmp[ 2 ], p2/2-tmp[ 1 ] ]
  else
    error "invalid argument: 
  fi
end:

# interval exponent
`Evalr/exp` := proc( )
  local a, inv:
  if has( [args], { [ ] } ) then return [ ] fi:
  if nargs = 1 then
    inv := args[ 1 ]:
    if type( inv, rat_ext ) then
      [ rfulb0( exp( inv ), 'r', 'l' ), rfulb0( exp( inv ), 'r', 'u' ) ]
    elif inv = [ ] then
      [ ]
    elif type( inv, 'interval' ) then
      [ rfulb0( exp( inv[ 1 ] ), 'r', 'l' ),
      rfulb0( exp( inv[ 2 ] ), 'r', 'u' ) ]
    else
      error "invalid argument: 
    fi
  elif type( args[ 1 ], rational ) and args[ 1 ] > 0 and args[ 1 ] <> 1 then
    a := args[ 1 ]:
    inv := args[ 2 ]:
    if type( inv, rat_ext ) then
      [ rfulb0( a^inv, 'r', 'l' ), rfulb0( a^inv, 'r', 'u' ) ]
    elif inv = [ ] then
      [ ]
    elif a > 1 and type( inv, 'interval' ) then
      [ rfulb0( a^( inv[ 1 ] ), 'r', 'l' ),
      rfulb0( a^( inv[ 2 ] ), 'r', 'u' ) ]
    elif a < 1 and type( inv, 'interval' ) then
      [ rfulb0( a^( inv[ 2 ] ), 'r', 'l' ),
      rfulb0( a^( inv[ 1 ] ), 'r', 'u' ) ]
    else
      error "invalid argument: 
    fi
  else
    error "invalid argument: 
  fi
end:

# interval logarithm
`Evalr/ln` := proc( )
  local a, inv:
  if has( [args], { [ ] } ) then return [ ] fi:
  if nargs = 1 then
    inv := args[ 1 ]:
    if type( inv, rat_ext ) and inv > 0 then
      # e.g. ln(4) is automatically simplified to 2*ln(2).
      [ rfulb0( 'ln'( inv ), 'r', 'l' ), rfulb0( 'ln'( inv ), 'r', 'u' ) ]
    elif inv = [ ] then
      [ ]
    elif type( inv, 'interval' ) and inv[ 1 ] >= 0 then
      [ `if`( inv[ 1 ] = 0, -infinity, rfulb0( 'ln'( inv[ 1 ] ), 'r', 'l' ) ),
      rfulb0( 'ln'( inv[ 2 ] ), 'r', 'u' ) ]
    else
      error "invalid argument: 
    fi
  elif type( args[ 1 ], rational ) and args[ 1 ] > 0 and args[ 1 ] <> 1 then
    a := args[ 1 ]:
    inv := args[ 2 ]:
    if type( inv, rat_ext ) and inv > 0 then
      [ rfulb0( 'ln'( inv ), 'r', 'l' )/
      rfulb0( 'ln'( a ), 'r', `if`( inv > 1, 'u', 'l' ) ),
      rfulb0( 'ln'( inv ), 'r', 'u' )/
      rfulb0( 'ln'( a ), 'r', `if`( inv > 1, 'l', 'u' ) ) ]
    elif inv = [ ] then
      [ ]
    elif a > 1 and type( inv, 'interval' ) and inv[ 1 ] >= 0 then
      [ `if`( inv[ 1 ] = 0, -infinity,
      rfulb0( 'ln'( inv[ 1 ] ), 'r', 'l' )/
      rfulb0( 'ln'( a ), 'r', `if`( inv[ 1 ] > 1, 'u', 'l' ) ) ),
      rfulb0( 'ln'( inv[ 2 ] ), 'r', 'u' )/
      rfulb0( 'ln'( a ), 'r', `if`( inv[ 2 ] > 1, 'l', 'u' ) ) ]
    elif a < 1 and type( inv, 'interval' ) and inv[ 1 ] >= 0 then
      [ rfulb0( 'ln'( inv[ 2 ] ), 'r', 'l' )/
      rfulb0( 'ln'( a ), 'r', `if`( inv[ 2 ] > 1, 'u', 'l' ) ),
      `if`( inv[ 1 ] = 0, infinity,
      rfulb0( 'ln'( inv[ 1 ] ), 'r', 'u' )/
      rfulb0( 'ln'( a ), 'r', `if`( inv[ 1 ] > 1, 'l', 'u' ) ) ) ]
    else
      error "invalid argument: 
    fi
  else
    error "invalid argument: 
  fi
end:

# interval power exponent
`Evalr/powexp` := proc( a::interval, x::interval )
  if a = [ ] or x = [ ] then
    [ ]
  elif a = 0 and x[ 1 ]*x[ 2 ] > 0 then
    [ 0, 0 ]
  elif a = 1 then
    [ 1, 1 ]
# elif a[ 1 ] >= 0 and `if`( a[ 1 ] = 0, x[ 1 ]*x[ 2 ] >= 0, true ) then
  elif a[ 1 ] >= 0 then
    `Evalr/exp`( `Evalr/multiply`( [ x, `Evalr/ln`( a ) ] ) )
  else
    error "invalid argument in 
  fi
end:

# to get an interval containing the number
`Evalr/shake` := proc( expr )
  local tm, tp:
  if Digits < 20 then Digits := 2*Digits: fi:
  tm := [ op( expr ) ]:
  tp := op( 0, expr ):
  if type( expr, rat_ext ) then
    [ expr, expr ]
  elif type( expr, Or( float, cons ) ) then
    [ rfulb0( expr, 'r', 'l' ), rfulb0( expr, 'r', 'u' ) ]
  elif tp = `+` then
    `Evalr/add`( map( procname, tm ) )
  elif tp = `*` then
    `Evalr/multiply`( map( procname, tm ) )
  elif tp = `^` then
    if type( tm[ 2 ], rational ) then
      `Evalr/power`( procname( tm[ 1 ] ), tm[ 2 ] )
    elif type( tm[ 1 ], rational ) then
      `Evalr/exp`( tm[ 1 ], procname( tm[ 2 ] ) )
    else
      `Evalr/powexp`( procname( tm[ 1 ] ), procname( tm[ 2 ] ) )
    fi
  elif tp in {`sin`,`cos`,`tan`,`cot`,
  `arcsin`,`arccos`,`arctan`,`arccot`,`exp`,`ln`} then
    thismodule[ cat( `Evalr/`,tp ) ]( procname( tm[ 1 ] ) )
  elif type( expr, Or( list, Matrix ) ) then
    map( procname, expr )
  else
    error "invalid argument: 
  fi
end:

# interval arithmetic
Evalr := proc( expr, invls::Or( list, set ) )
  local i, vars, t, tm, tp, tmp:
  option system, remember:
  vars := select( type, indets( expr ), name ):
  if vars = { } then
    if type( expr, consnum ) then
      return `Evalr/shake`( expr )
    elif has( expr, { [ ] } ) then
      return [ ]
    elif type( expr, 'interval' ) then
      return expr
    elif type( expr, list ) and nops( expr ) = 2 and
    type( expr[ 1 ], consnum ) and type( expr[ 2 ], consnum ) then
      return [ `Evalr/shake`( expr[ 1 ] )[ 1 ],
      `Evalr/shake`( expr[ 2 ] )[ 2 ] ]
    elif type( expr, Or( list, Matrix ) ) then
      return map( procname, expr, args[ 2.. -1 ] )
    elif not member( op( 0, expr ),
    {`+`, `*`, `^`, `sin`, `cos`, `tan`, `cot`,
    `arcsin`, `arccos`, `arctan`, `arccot`, `exp`, `ln`} ) then
      error "invalid argument: 
    fi:
  else
    if nargs = 1 then
      return procname( expr, map( t -> t = [ -infinity, infinity ], vars ) )
    elif member( [ ], subs( invls, vars ) ) then
      return [ ]
    elif vars minus { op( map( lhs, invls ) ) } <> { } then
      return procname( expr,
      { op( invls ) } union map( t -> t = [ -infinity, infinity ],
      vars minus { op( map( lhs, invls ) ) } ) )
    elif not type( subs( invls, vars ), set( 'interval' ) ) then
      tmp := [ ]:
      t := [ op( map( t -> t = subs( invls, t ), vars ) ) ]:
      for i to nops( t ) do
        if type( rhs( t[ i ] ), 'interval' ) then
          tmp := [ op( tmp ), t[ i ] ]:
        elif type( rhs( t[ i ] ), list ) and nops( rhs( t[ i ] ) ) = 2 and
        type( rhs( t[ i ] )[ 1 ], consnum ) and
        type( rhs( t[ i ] )[ 2 ], consnum ) then
          tmp := [ op( tmp ), lhs( t[ i ] ) =
          [ `Evalr/shake`( rhs( t[ i ] )[ 1 ] )[ 1 ],
          `Evalr/shake`( rhs( t[ i ] )[ 2 ] )[ 2 ] ] ]:
        elif type( rhs( t[ i ] ), consnum ) then
          tmp := [ op( tmp ), `Evalr/shake`( rhs( t[ i ] ) ) ]:
        else
          error "invalid evaluation range: 
        fi:
      od:
      return procname( expr, tmp )
    elif type( expr, name ) then
      return subs( invls, expr )
    fi:
  fi:

  tm := [ op( expr ) ]:
  tp := op( 0, expr ):
  if tp = `+` then
    `Evalr/add`( map( procname, tm, args[ 2..-1 ] ) )
  elif tp = `*` then
    `Evalr/multiply`( map( procname, tm, args[ 2..-1 ] ) )
  elif tp = `^` then
    if type( tm[ 2 ], rational ) then
      `Evalr/power`( procname( tm[ 1 ], args[ 2..-1 ] ), tm[ 2 ] )
    elif type( tm[ 1 ], rational ) then
      `Evalr/exp`( tm[ 1 ], procname( tm[ 2 ], args[ 2..-1 ] ) )
    else
      `Evalr/powexp`( procname( tm[ 1 ], args[ 2..-1 ] ),
      procname( tm[ 2 ], args[ 2..-1 ] ) )
    fi
  elif tp in {`sin`, `cos`, `tan`, `cot`,
  `arcsin`, `arccos`, `arctan`, `arccot`, `exp`, `ln`} then
    thismodule[ cat( `Evalr/`,tp ) ]( procname( tm[ 1 ], args[ 2..-1 ] ) )
  elif type( expr, Or( list, Matrix ) ) then
    map( procname, expr, args[ 2.. -1 ] )
  else
    error "invalid argument: 
  fi
end:

##############################################################################
# Auxiliary procedures
##############################################################################

# list sorting
sortpos := proc( ls::list )
  local v, i, j, n, tmp:
  n := nops( ls ):
  v := [ seq( [ ls[ i ], i ], i = 1..n ) ]:
  for i from 2 to n do
    if ( v[ i ][ 1 ] ) < ( v[ i-1 ][ 1 ] ) then
      tmp := v[ i ]:
      v[ i ] := v[ i-1 ]:
      for j from i-2 by -1 to 1 while ( tmp[ 1 ] ) < ( v[ j ][ 1 ] ) do
        v[ j+1 ] := v[ j ]:
      od:
    if j = 1 and ( tmp[ 1 ] ) < ( v[ j ][ 1 ] ) then v[ j ] := tmp:
    else v[ j+1 ] := tmp:
    fi:
    fi:
  od:
  return map2( op, 2, v ):
end:

# The "isdef" procedure determines whether or not a interval matrix is
# positive definite
# input: interval matrix and options, can receive 1 to 3 parameters
# interval matrix can have following forms:
## 1.( [ a, b ] )_{nn}, one parameter
## 2.[ A, B ], two parameters
# option query=posdef/negdef:
## determine positive (default)/negative definiteness.
# option query=possemidef/negsemidef:
## determine positive /negative semidefiniteness.
# option query=nonposdef/nonnegdef:
## determine nonpositive/nonnegative definiteness.
# option query=nonpossemidef/nonnegsemidef:
## determine nonpositive/nonnegative semidefiniteness.
# option query=nondef:
## determine nondefiniteness.
# option method=vertex/eigenvalue:
## use different methods to determine positive /negative definiteness:
## the method using vertex matrices(a sufficient and necessary condition,
## the default method), and the method using eigenvalues(a sufficient condition)
## in which case the "false" returned just indicates an unsuccessful test.
# output: true/false
# Note: in order to reduce the time of datatype checking,
# the entries in matrices should have the same datetype.
isdef := proc()
  local A, B, n, _query, _method, i, j, k, q, S, v, Ac, Ad, Ace, Ade, temp, x:
  uses LinearAlgebra:
  A := 0: B := 0: n := 0: _query := 0: _method := 0:
  for i to nargs do
    if type( args[ i ], 'Matrix'( square ) ) then

      # more accurate form should be:
      # if type( args[ i ], 'Matrix'( interval ) ) then
      if type( args[ i ][ 1, 1 ], 'interval' ) then
        A := map2( op, 1, args[ i ] ):
        B := map2( op, 2, args[ i ] ):
        n := RowDimension( A ):

      # more accurate form should be:
      # elif type( args[ i ], 'Matrix'( constant ) ) then
      # and comparison should be contained.
      elif type( args[ i ][ 1, 1 ], rat_ext ) then
        if A = 0 then
          A := args[ i ]:
          n := RowDimension( A ):
        else
          if RowDimension( args[ i ] ) <> n then
            error "different matrix dimension: 
          fi:
          B := args[ i ]:
        fi:
      else error "error matrix: 
      fi:
    elif type( args[ i ],
    identical( query ) =
    { identical( posdef ), identical( negdef ),
    identical( possemidef ), identical( negsemidef ),
    identical( nonposdef ), identical( nonnegdef ),
    identical( nonpossemidef ), identical( nonnegsemidef ),
    identical( nondef ) } ) then
      _query := rhs( args[ i ] ):
    elif type( args[ i ], identical( method ) =
    { identical( vertex ),identical( eigenvalue ) } ) then
      _method := rhs( args[ i ] ):
    else error "invalid argument: 
    fi:
  od:

  if A = 0 or B = 0 or n = 0 then error "invalid arguments: 
  if has( A, infinity ) or has( B, infinity ) then return false fi:
  if _query = 0 then _query := posdef fi:
  if member( _query,
  { 'posdef', 'negdef', 'possemidef', 'negsemidef' } ) and
  _method = 0 then
    _method := vertex
  fi:
  if member( _query,
  { nonposdef, nonnegdef, nonpossemidef, nonnegsemidef, nondef } ) and
  _method = 'vertex' then
    error "the method specified cannot be applied"
  fi:

  if _method = 'vertex' then
    q := 'positive_definite';
    if _query = 'negdef' or _query = 'negsemidef' then
      temp := A: A := -B: B := -temp:
    fi:
    if _query = 'possemidef' or _query = 'negsemidef' then
      q := 'positive_semidefinite';
    fi;
    for k from 0 to 2^( n-1 )-1 do
      v := Vector( n, convert( k, base, 2 ) ):
      S := Matrix( n, n, ( i, j ) ->
      ( ( A[ i, j ]+B[ i, j ] )+
      ( -1 )^( v[ i ]+v[ j ] )*( A[ i, j ]-B[ i, j ] ) )/2,
      shape = symmetric ):
      if LA_Main:-IsDefinite( S, 'query' = q ) = false then
        return false
      fi:
    od:
    return true
  else
    Ac := Matrix( n, n, ( i, j ) -> ( A[ i, j ]+B[ i, j ] )/2,
    shape = symmetric ):
    Ad := Matrix( n, n, ( i, j ) -> ( -A[ i, j ]+B[ i, j ] )/2,
    shape = symmetric ):
    Ace := realroot( CharacteristicPolynomial( Ac, x ), 10^( -Digits ) ):
    Ade := realroot( CharacteristicPolynomial( Ad, x ), 10^( -Digits ) ):
    Ade := max( op( map( abs, map( op, Ade ) ) ) ):
    if _query = 'posdef' and
    min( op( map2( op, 1, Ace ) ) )-Ade > 0 then
      return true
    elif _query = 'negdef' and
    max( op( map2( op, 2, Ace ) ) )+Ade < 0 then
      return true
    elif _query = 'possemidef' and
    min( op( map2( op, 1, Ace ) ) )-Ade >= 0 then
      return true
    elif _query = 'negsemidef' and
    max( op( map2( op, 2, Ace ) ) )+Ade <= 0 then
      return true
    elif _query = 'nonposdef' and
    min( op( map2( op, 1, Ace ) ) )+Ade <= 0 then
      return true
    elif _query = 'nonnegdef' and
    max( op( map2( op, 2, Ace ) ) )-Ade >= 0 then
      return true
    elif _query = 'nonpossemidef' and
    min( op( map2( op, 1, Ace ) ) )+Ade < 0 then
      return true
    elif _query = 'nonnegsemidef' and
    max( op( map2( op, 2, Ace ) ) )-Ade > 0 then
      return true
    elif min( op( map2( op, 1, Ace ) ) )+Ade < 0 and
    max( op( map2( op, 2, Ace ) ) )-Ade > 0 then
      return true
    else
      return false
    fi:
  fi:
end:

# get vertexes of a domain
vert := proc( va::list )
  local n,i,j,k,vt:
  n := nops( va ):
  vt := []:
  if type( va[ 1 ], `=` ) then
    for i from 0 to 2^n-1 do
      j := Vector( n, convert( i, base, 2 ) ):
      vt := [ op( vt ), [ seq( lhs( va[ k ] )
      =op( j[ k ]+1, rhs( va[ k ] ) ), k=1..n ) ] ]:
    od:
  elif type( va[ 1 ], 'interval' ) then
    for i from 0 to 2^n-1 do
      j := Vector( n, convert( i, base, 2 ) ):
      vt := [ op( vt ), [ seq( op( j[ k ]+1, va[ k ] ), k=1..n ) ] ]:
    od:
  fi:
  return vt:
end:

# get the convex hull of a set of intervals
inthull := proc( invls::list(interval) )
  if nops( invls ) = 1 then
    invls[ 1 ]
  else
    [ op( sortpos( map2( op, 1, invls ) )[ 1 ], invls )[ 1 ],
    op( sortpos( map2( op, 2, invls ) )[ -1 ], invls )[ 2 ] ]
  fi
end:

# get the width of an interval
intwidth := proc( inv )
  if type( inv, 'interval' ) then
    inv[ 2 ] - inv[ 1 ]
  elif type( inv, rational ) then
    0
  elif type( inv, `=` ) and type( rhs( inv ), 'interval' ) then
    rhs( inv )[ 2 ] - rhs( inv )[ 1 ]
  else
    error "invalid argument: 
  fi
end:#

# get the position of the interval with the maximal width in
# a list of intervals
maxwidthdim := proc( invls::list )
  op( -1, sortpos( map( intwidth, invls ) ) )
end:

# subdivides a domain over an interval or all intervals
intsbdv := proc( invls::list, sbdv::integer )
  local i, j, dim, s, t, n, cnt, subint:
  n := nops( invls ):
  if sbdv = 1 or has( invls, infinity ) then return [ invls ] fi:
  dim := 0:
  if nargs > 2 and type( args[ 3 ], integer ) and args[ 3 ] <= n then
    dim := args[ 3 ]
  elif nargs > 2 then
    error "invalid argument: 
  fi:
  if dim <> 0 then
    if type( invls[ 1 ], `=` ) then
      return [ seq( [ op( invls[ 1..dim-1 ] ), lhs( invls[ dim ] )=
             [ rhs( invls[ dim ] )[ 1 ] +
             i/sbdv*intwidth( rhs( invls[ dim ] ) ),
             rhs( invls[ dim ] )[ 1 ] +
             ( i+1 )/sbdv*intwidth( rhs( invls[ dim ] ) ) ],
             op( invls[ dim+1..n ] ) ], i = 0..sbdv-1 ) ]:
    elif type( invls[ 1 ], 'interval' ) then
      return [ seq( [ op( invls[ 1..dim-1 ] ),
             [ invls[ dim, 1 ] + i/sbdv*intwidth( invls[ dim ] ),
             invls[ dim, 1 ] + ( i+1 )/sbdv*intwidth( invls[ dim ] ) ],
             op( invls[ dim+1..n ] ) ], i = 0..sbdv-1 ) ]:
    fi:
  fi:
  cnt := sbdv^n-1:
  subint := []:
  if type( invls[ 1 ], `=` ) then
    for i from 0 to cnt do
      j := Vector[ row ]( n,convert( i,base,sbdv ) ):
      subint := [ op( subint ),
              [ seq( lhs( invls[ s ] ) =
              [ rhs( invls[ s ] )[ 1 ]+
              j[ s ]/sbdv*( rhs( invls[ s ] )[ 2 ]-rhs( invls[ s ] )[ 1 ] ),
              rhs( invls[ s ] )[ 1 ]+
              ( j[ s ]+1 )/sbdv*( rhs( invls[ s ] )[ 2 ]-
              rhs( invls[ s ] )[ 1 ] ) ],
              s=1..n ) ] ]:
    od
  elif type( invls[ 1 ], 'interval' ) then
    for i from 0 to cnt do
      j := Vector[ row ]( n, convert( i, base, sbdv ) ):
      subint := [ op( subint ),
              [ seq(
              [ invls[ s ][ 1 ]+
              j[ s ]/sbdv*( invls[ s ][ 2 ]-invls[ s ][ 1 ] ),
              invls[ s ][ 1 ]+
              ( j[ s ]+1 )/sbdv*( invls[ s ][ 2 ]-invls[ s ][ 1 ] ) ],
              s=1..n ) ] ]:
    od
  fi:
  return subint
end:

# test whether a rational number or an interval is contained in an interval
contain := proc( inv::Or( name = interval, interval ),
p::Or( name = Or( rational, interval ), rational, interval ) )
  if type( inv, name = 'interval' ) then
    procname( rhs( inv ), p )
  elif type( p, name = Or( rational, 'interval' ) ) then
    procname( inv, rhs( p ) )
  elif type( p, rational ) then
    evalb( inv[ 1 ] <= p and p <= inv[ 2 ] )
  elif type( p, 'interval' ) then
    evalb( inv[ 1 ] <= p[ 1 ] and p[ 1 ] <= inv[ 2 ] and
    inv[ 1 ] <= p[ 2 ] and p[ 2 ] <= inv[ 2 ] )
  else
    error "invalid arguments: 
  fi
end:

end module:
\end{maplettyout}
\clearpage

\subsection{Module fivepoints} \label{fivepoints}
\begin{maplettyout}
# fivepoints: a Maple package used for solving the problem of spherical
# distribution of 5 points.
# $revision: 1.0$
fivepoints := module()
  export spf, cmb, changecoord, eucdis, grad, normalvector,
    ischecked, spchecked,
    howchecked, inwhichdomain, subdivisiondomain, subdivisionprocess:
  local init:
  option package, load = init:

init := proc( )
local p, q, i, j, k, m, nv, temp, x, y, A, B, C, D, E, py, pyd, ls:
global truncatenegativepart,
  phi, theta, indet, Vt, Vtd, f, F, Fls, Df, H,
  Thetabp, fmax, fmaxrat, pyang, Thetap, pyfmax, xi, chi, eta, zeta, bpyva, pyva,
  varsrange1, varsrange2, ckrangels1, ckrangels2,
  gr, grcv, signdis, Methods:
uses IntervalArithmetic:

# Negative parts of base intervals in interval power arithmetic are truncated.
# e.g. Evalr( op( 6, f ), ckrangels1[ 1 ] );
truncatenegativepart := true:

if FileTools[ Exists ]( "fivepoints.data" ) = false then

# spherical coordinates of five points
Vt := [ [ 1, 0, 0 ], [ 1, phi[ 1 ], Pi ], [ 1, phi[ 2 ], theta[ 2 ] ],
[ 1, phi[ 3 ], theta[ 3 ] ], [ 1, phi[ 4 ], theta[ 4 ] ] ]:

# Cartesian coordinates corresponding the spherical coordinates
Vtd := map( changecoord, Vt, s2d ):

# indeterminates
indet := [ phi[ 1 ], op( map( op, [ seq( [ phi[ i ], theta[ i ] ],
i = 2..4 ) ] ) ) ]:

# sum of mutual distances
f := add( add( sqrt( ( cos( Vt[ i, 2 ] )*cos( Vt[ i, 3 ] )-
cos( Vt[ j, 2 ] )*cos( Vt[ j, 3 ] ) )^2+
( cos( Vt[ i, 2 ] )*sin( Vt[ i, 3 ] )-
cos( Vt[ j, 2 ] )*sin( Vt[ j, 3 ] ) )^2+
( sin( Vt[ i, 2 ] )-sin( Vt[ j, 2 ] ) )^2 ), j = i+1..5 ), i = 1..4 ):
f := spf( simplify( f ) ):

Fls := [ seq( op( i, f ), i = 1..10 ) ]:
p := [ A, B, C, D, E ]:
q := [ seq( seq( [ i, j ], j = i+1..5 ), i = 1..4 ) ]:
# ten distances
F := table( [ seq( cat( p[ q[ i ][ 1 ] ], p[ q[ i ][ 2 ] ] ) = Fls[ i ],
i = 1..10 ) ] ):

# list of derivatives of f
Df := [ seq( diff( f, indet[ i ] ), i = 1..7 ) ]:
# Hessian matrix
H := Matrix( 7, 7, ( i, j ) -> ( diff( f, indet[ i ], indet[ j ] ) ) ):

# coordinate corresponding the bipyramid distribution of 5 points
Thetabp := [ phi[ 1 ] = -1/3*Pi, phi[ 2 ] = 1/3*Pi, theta[ 2 ] = Pi,
phi[ 3 ] = 0, theta[ 3 ] = -1/2*Pi, phi[ 4 ] = 0, theta[ 4 ] = 1/2*Pi ]:
# distance sum corresponding the bipyramid distribution of 5 points
fmax := simplify( subs( Thetabp, f ) ): fmaxrat := rfulb( fmax, 'r', 'l' ):

py := ( 4+4*sqrt( 2 ) )*sqrt( 1-tt^2 )+4*sqrt( 2 )*sqrt( 1+tt ):
pyd := diff( py, tt ):
# angle between the line connecting the spherical center and
# the vertex of the pyramid bottom, and the pyramid bottom
pyang := arcsin( solve( pyd, tt ) ):
# coordinate corresponding the pyramid distribution of 5 points
Thetap := [ phi[ 1 ] = -2*pyang, phi[ 2 ] = Pi/2-pyang, theta[ 2 ] = Pi,
phi[ 3 ] = -arcsin( sin( pyang )*cos( pyang ) ),
theta[ 3 ] = -arccot( sin( pyang )^2/cos( pyang ) ),
phi[ 4 ] = -arcsin( sin( pyang )*cos( pyang ) ),
theta[ 4 ] = arccot( sin( pyang )^2/cos( pyang ) ) ]:
# distance sum corresponding the pyramid distribution of 5 points
pyfmax := subs( Thetap, f ):

# radius of the domain first excluded near coordinates corresponding
# the bipyramid distribution
xi := 1/377*Pi:
# radius of the domain first excluded near coordinates corresponding
# the pyramid distribution
chi := 1/791*Pi:
# radius of the domain excluded near coordinates corresponding
# the bipyramid distribution
eta := xi:
# radius of the domain excluded near coordinates corresponding
# the pyramid distribution
zeta := chi:
bpyva := [ seq( indet[ i ] = Evalr( rhs( Thetabp[ i ] ) + [ -eta, eta ] ),
i = 1..7 ) ]:
pyva := [ seq( indet[ i ] = Evalr( rhs( Thetap[ i ] ) + [ -zeta, zeta ] ),
i = 1..7 ) ]:

# one domain need to be checked
varsrange1 := [ phi[ 1 ] = [ -2*arccos( ( fmax-2 )/9/2 ), 0 ],
phi[ 2 ] = [ 0, 1/2*Pi ], theta[ 2 ] = [ 0, Pi ],
phi[ 3 ] = [ -1/2*Pi, 1/2*Pi ], theta[ 3 ] = [ -Pi, 0 ],
phi[ 4 ] = [ -1/2*Pi, 1/2*Pi ], theta[ 4 ] = [ 0, Pi ] ]:
varsrange1 := map( x -> lhs( x ) = Evalr( rhs( x ) ), varsrange1 ):
# another domain need to be checked
varsrange2 := [ phi[ 1 ] = [ -2*arccos( ( fmax-2 )/9/2 ), 0 ],
phi[ 2 ] = [ -1/2*Pi, 0 ], theta[ 2 ] = [ 0, Pi ],
phi[ 3 ] = [ 0, 1/2*Pi ], theta[ 3 ] = [ -Pi, 0 ],
phi[ 4 ] = [ -1/2*Pi, 0 ], theta[ 4 ] = [ 0, Pi ] ]:
varsrange2 := map( x -> lhs( x ) = Evalr( rhs( x ) ), varsrange2 ):

k := 3:
# subdivide the first domain
ckrangels1 := intsbdv( varsrange1, k ):
# subdivide the second domain
ckrangels2 := intsbdv( varsrange2, k ):

# used to test whether or not 5 points are on the same half sphere
signdis := []:
for i from 1 to 4 do
  for j from `if`( i = 1, 3, i+1 ) to 5 do
    temp := []:
    nv := normalvector( [ Vtd[ i ], Vtd[ j ] ] ):
    for k from 1 to 5 do
      if k = i or k = j then next: fi:
      temp := [ op( temp ), add( nv[ m ]*Vtd[ k ][ m ], m = 1..3 ) ]:
    od:
    signdis := [ op( signdis ), temp ]:
  od:
od:

# gradients at points
gr := [ seq( grad( i ), i = 1..5 ) ]:
grcv := [ -gr[ 1, 2 ], -gr[ 1, 3 ], -gr[ 2, 2 ],
gr[ 2, 1 ]*Vtd[ 2, 3 ]-gr[ 2, 3 ]*Vtd[ 2, 1 ],
seq( gr[ i, 1 ]*Vtd[ i, 2 ]-gr[ i, 2 ]*Vtd[ i, 1 ], i = 3..5 ),
seq( gr[ i, 1 ]*Vtd[ i, 3 ]-gr[ i, 3 ]*Vtd[ i, 1 ], i = 3..5 ) ]:
# a condition used to test whether there exists a maximal distribution
# in a domain
grcv := spf( cmb( expand( grcv ) ) ):

# method names used to test
Methods := [ defver,defeig, nondef, gradient, mindis, secondlength,
  dis, derivative, secordder, pointc, halfsphere ]:

# save values of these variables in the file "fivepoints.data"
# if it does not exist
save indet, Vt, Vtd, f, F, Fls, Df, H,
  Thetabp, fmax, fmaxrat, pyang, Thetap, pyfmax,
  xi, chi, eta, zeta, bpyva, pyva,
  varsrange1, varsrange2, ckrangels1, ckrangels2,
  gr, grcv, signdis, Methods, "fivepoints.data":

else
# read values of values from the file "fivepoints.data"
  read "fivepoints.data":

fi:

# generate the process of subdisivion
subdivisionprocess():

end:

# simplify expressions so as to raise the efficiency of interval computation
spf  :=  proc( expr )
  local x:
  if type( expr,  `+` ) and hastype( expr,  radical ) then
    `+`( op( procname( [ op( expr ) ] ) ) )
  elif type( expr,  `+` ) then
    map( combine,
    factor( `+`( op( select( x -> nops( x ) = 5, [ op( expr ) ] ) ) ) ) )+
    `+`( op( select( x -> nops( x ) <> 5, [ op( expr ) ] ) ) )
  elif type( expr,  `*` ) and hastype( expr,  radical ) then
    `*`( op( procname( [ op( expr ) ] ) ) )
  elif type( expr,  Or( `*`,  trig,  arctrig,  constant ) ) then
    expr
  elif type( expr,  `^` ) then
    `^`( procname( op( 1,  expr ) ), op( 2, expr ) )
  elif type( expr,  list ) then
    map( procname,  expr )
  else
    error "unrecognized expression type: 
  fi
end:

# combine terms in denominators
cmb := proc( expr )
  local tm, tmp, x, i:
  if type( expr, list ) then return map( procname, expr )
  elif type( expr, Not( `+` ) ) then return expr
  fi:
  tm := [ op( expr )]: tmp := 0:
  for i to nops( tm ) do
    if tmp = 0 then
      tmp := [ tm[ i ] ]
    else
      tmp := selectremove( x -> denom( x ) = denom( tm[ i ] ), tmp ):
      if tmp[ 1 ] = [] then
        tmp := [ op( tmp[ 2 ] ), tm[ i ] ]
      else
        tmp := [ op( tmp[ 2 ] ),
        ( numer( op( tmp[ 1 ] ) ) + numer( tm[ i ] ) )/denom( tm[i] ) ]
      fi
    fi
  od:
  convert( tmp, `+` )
end:

# convert from one coordinate system to another
# s2d: from spherical system to Cartesian system
# d2s: from Cartesian system to spherical system
# Cartesian system: [ x1, x2, x3 ]
# spherical system: [ r, phi, theta ]
# ( r >= 0, -Pi/2 =< phi <= Pi/2, -Pi < theta <= Pi )
changecoord := proc( L::list, opt )
  uses IntervalArithmetic:
  if nargs = 1 then return changecoord( L, s2d ): fi:
  if type( opt, identical( s2d ) ) then
    return [ L[ 1 ]*cos( L[ 2 ] )*cos( L[ 3 ] ),
    L[ 1 ]*cos( L[ 2 ] )*sin( L[ 3 ] ), L[ 1 ]*sin( L[ 2 ] ) ]:
  elif type( opt, identical( d2s ) ) then
    if L[ 1 ] = 0 and L[ 2 ] = 0 and L[ 3 ] > 0 then
      return [ L[ 3 ], Pi/2, 0 ]:
    elif L[ 1 ] = 0 and L[ 2 ] = 0 and L[ 3 ] < 0 then
      return [ -L[ 3 ], -Pi/2, 0 ]:
    elif L[ 1 ] = 0 and L[ 2 ] = 0 then
      return [ 0, 0, 0 ]:
    elif L[ 1 ] > 0 and L[ 2 ] = 0 and L[ 3 ] = 0 then
      return [ L[ 1 ], 0, 0 ]:
    elif L[ 1 ] < 0 and L[ 2 ] = 0 and L[ 3 ] = 0 then
      return [ L[ 1 ], 0, Pi ]:
    elif L[ 1 ] > 0 then
      return [ sqrt( L[ 1 ]^2+L[ 2 ]^2+L[ 3 ]^2 ),
      arctan( L[ 3 ]/sqrt( L[ 1 ]^2+L[ 2 ]^2 ) ),
      arctan( L[ 2 ]/L[ 1 ] ) ]:
    elif L[ 1 ] < 0 and L[ 2 ] >= 0 then
      return [ sqrt( L[ 1 ]^2+L[ 2 ]^2+L[ 3 ]^2 ),
      arctan( L[ 3 ]/sqrt( L[ 1 ]^2+L[ 2 ]^2 ) ),
      arctan( L[ 2 ]/L[ 1 ] )+Pi ]:
    elif L[ 1 ] < 0 and L[ 2 ] < 0 then
      return [ sqrt( L[ 1 ]^2+L[ 2 ]^2+L[ 3 ]^2 ),
      arctan( L[ 3 ]/sqrt( L[ 1 ]^2+L[ 2 ]^2 ) ),
      arctan( L[ 2 ]/L[ 1 ] )-Pi ]:
    fi:
  fi:
end:

# Euclidean distances
eucdis := proc()
  local i:
  if nargs = 3 then
    if type( args[ 3 ], identical( sphere ) ) then
      procname( changecoord( args[ 1 ] ), changecoord( args[ 2 ] ) )
    elif type( args[ 3 ], identical( descartes ) ) then
      procname( args[ 1 ], args[ 2 ] )
    else
      error "invalid argument: 
    fi:
  elif nargs = 2 and type( args[ 1 ], list ) and type( args[ 2 ], list ) then
    sqrt( sum( ( args[ 1 ][ i ]-args[ 2 ][ i ] )^2, i = 1..3 ) )
  elif nargs = 2 and type( args[ 1 ], list ) and
  type( args[ 2 ], identical( sphere ) ) then
    eucdis( changecoord( args[ 1 ] ) )
  elif ( nargs = 2 and type( args[ 1 ], list ) and
  type( args[ 2 ], identical( descartes ) ) ) or nargs = 1 then
    sqrt( sum( ( args[ 1 ][ i ] )^2, i = 1..3 ) )
  else
    error "invalid argument: 
  fi:
end:

# gradients
grad := proc( n::integer )
  local i:
  global Vtd:
  add( map( `*`, Vtd[ n ]-Vtd[ i ],
  1/simplify( eucdis( Vtd[ n ], Vtd[ i ] ) ) ),
  i in [ $1..( n-1 ), $( n+1 )..5 ] ):
end:

# normal vector of the plane determined by two points and
# the spherical center
normalvector := proc( L::listlist )
  [ L[ 1, 2 ]*L[ 2, 3 ]-L[ 2, 2 ]*L[ 1, 3 ],
  L[ 1, 3 ]*L[ 2, 1 ]-L[ 2, 3 ]*L[ 1, 1 ],
  L[ 1, 1 ]*L[ 2, 2 ]-L[ 1, 2 ]*L[ 2, 1 ] ]
end:

# use various methods to check a rectangular domain
ischecked := proc( invls )
  local _method, i, j, intm, x, y, z, df, iv, tmp:
  global H, Fls, gr, f, fmaxrat, Df, Digits, Vtd, grcv, signdis:
  uses IntervalArithmetic:
  _method := 0: intm := 0:
  if nargs > 1 then
    for i from 2 to nargs do
      if type( args[ i ], identical( method ) ={

      # use vertex matrices to determine the negative definiteness of
      # the interval Hessian matrix
      identical( defver ),

      # use eigenvalues to determine the negative definiteness of
      # the interval Hessian matrix
      identical( defeig ),

      # use eigenvalues to determine the nonnegative definiteness of
      # the interval Hessian matrix
      identical( nondef ),

      # use gradients to determine there exists no maximum in the domain
      identical( gradient ),

      # check that the distance of some two points is less then
      # a certain value
      identical( mindis ),

      # check that the distance of A and B is not the second longest
      identical( secondlength ),

      # use interval computation to show directly that
      # the upper bound of f in the domain is less then fmax
      identical( totaldis ),

      # use interval computation to show directly that
      # one of the derivatives doesnot change siges in the domain
      identical( derivative ),

      # analyze second order derivatives to determine that
      # one of the derivatives doesnot change siges in the domain
      identical( secordder ),

      # check that C is below E, which contradicts assumptions
      identical( pointc ),

      # check that all points are on the same half sphere
      # in which case f couldnot obtain its maximum
      identical( halfsphere )
      } )
      then
        _method := rhs( args[ i ] ):
      elif type( args[ i ], 'Matrix'( square ) ) then
        intm := args[ i ]:
      else error "unrecognized method:
      fi:
    od:
  fi:

  if _method = 0 then _method := totaldis: fi:
  if member( _method, {defver, defeig, nondef} ) then
    if intm = 0 then
      intm := Matrix( 7, 7, ( i, j ) -> Evalr( H[ i, j ], invls ),
      shape = symmetric ):
    fi:
    if _method = defver then
      isdef( intm, method = vertex, query = negdef )
    elif _method = defeig then
      isdef( intm, method = eigenvalue, query = negdef )
    else
      isdef( intm, query = nonnegsemidef )
    fi:
  elif _method = gradient then
    for i to nops( grcv ) do
      x := Evalr( grcv[ i ], invls ):
      if x[ 1 ] > 0 or x[ 2 ] < 0 then
        return true:
      fi:
    od:
    return false:
  elif _method = mindis then
    for i to nops( Fls ) do
      if Evalr( Fls[ i ], invls )[ 2 ] < ( 2/15 ) then return true: fi:
    od:
    return false:
  elif _method = secondlength then
    x := 0:
    y := Evalr( Fls[ 1 ], invls ):
    z := []:
    for i from 2 to nops( Fls ) do
      z := [ op( z ), Evalr( Fls[ i ], invls ) ]:
      if z[ -1 ][ 1 ] > y[ 2 ] then
        if x = 0 then x := 1:
        else return true:
        fi:
      fi:
    od:
    if member( false, map( tmp -> evalb( tmp[ 2 ] < y[ 1 ] ), z ) ) = false
    then
      return true:
    else
      return false:
    fi:
  elif _method = totaldis then
    return evalb( Evalr( f, invls )[ 2 ] < fmaxrat ):
  elif _method = derivative then
    for i to 7 do
      df := Evalr( Df[ i ], invls ):
      if df[ 1 ] > 0 or df[ 2 ] < 0 then return true: fi:
    od:
    return false:
  elif _method = secordder then
    if intm = 0 then
      intm := Matrix( 7, 7, ( i, j ) -> Evalr( H[ i, j ], invls ),
      shape = symmetric ):
    fi:
    df := []:
    for i to 7 do
      iv := []:
      for j to 7 do
        if intm[ i, j ][ 1 ] > 0 then
          iv := [ op( iv ),
          [ lhs( invls[ j ] ) = rhs( invls[ j ] )[ 1 ],
          lhs( invls[ j ] ) = rhs( invls[ j ] )[ 2 ] ] ]:
        elif intm[ i, j ][ 2 ] < 0 then
          iv := [ op( iv ),
          [ lhs( invls[ j ] ) = rhs( invls[ j ] )[ 2 ],
          lhs( invls[ j ] ) = rhs( invls[ j ] )[ 1 ] ] ]:
        fi:
        iv := [ op( iv ), [ invls[ j ], invls[ j ] ] ]:
      od:
      x := Evalr( Df[ i ], map2( op, 1, iv ) ):
      y := Evalr( Df[ i ], map2( op, 2, iv ) ):
      if type( x, constant ) then :
      else x := x[ 1 ]:
      fi:
      if type( y, constant ) then :
      else y := y[ 2 ]:
      fi:
      df := [ op( df ), [ x, y ] ]:
      if df[ -1 ][ 1 ] > 0 or df[ -1 ][ 2 ] < 0 then return true: fi:
    od:
    return false:
  elif _method = pointc then
    if subs( invls, phi[ 2 ] )[ 2 ] - subs( invls, phi[ 4 ] )[ 1 ] < 0 then
      return true:
    else
      return false:
    fi:
  elif _method = halfsphere then
    for i to nops( signdis ) do
      x := 0: tmp := true:
      for j to 3 do
        y := Evalr( signdis[ i, j ], invls ):
        if type( y, constant ) then
          if y > 0 then y := 1:
          elif y < 0 then y := -1:
          else y := 0:
          fi:
        else
          if y[ 1 ] > 0 then y := 1:
          elif y[ 2 ] < 0 then y := -1:
          else y := 0:
          fi:
        fi:
        if y = 0 then tmp := false: break:
        elif y <> 0 and x = 0 then x := y:
        elif y <> 0 and x <> y then tmp := false: break:
        fi:
      od:
      if tmp then return true: fi:
    od:
    return false:
  fi:
end:

# use specified methods to check a rectangular domain,
# if not successful, the domain is subdivided,
# and then check subdomains recursively,
# until maximal width of subdomains is less than some value.
spchecked := proc( va::list )
  local x, y, methods, dim, subint, i, cur, bpyck, pyck, tmp:
  # "notchecked" and "checkprocess" should be set to [] before the procedure
  # is first called, since they store domains cannot be checked successfully,
  # and the process of checking.
  global bpyva, pyva, notchecked, checkprocess:
  uses IntervalArithmetic:

  methods := [ totaldis ]:
  # The case when 5 points form a distribution close to a bipyramid or
  # a pyramid should be checked separately.
  bpyck := true: pyck := true:
  if nargs > 1 then
    for i from 2 to nargs do
      if type( args[ i ], identical( notcheckbipyramid ) ) then
        bpyck := false:
      elif type( args[ i ], identical( notcheckpyramid ) ) then
        pyck := false:
      elif type( args[ i ], list ) then
        methods := args[ i ]:
      else
        error "invalid argument: 
      fi:
    od:
  fi:

  if bpyck then
    for i to 7 do
      if not contain( rhs( bpyva[ i ] ), rhs( va[ i ] ) ) then break fi
    od:
    if i > 7 then
      checkprocess := [ op( checkprocess ), [ -1 ] ]:
      return true:
    fi:
    for i to 7 do
      if rhs( va[ i ] )[ 2 ] < rhs( bpyva[ i ] )[ 1 ]
      or rhs( va[ i ] )[ 1 ] > rhs( bpyva[ i ] )[ 2 ] then
        break
      fi
    od:
    if i <= 7 then
      bpyck := false
    fi:
  fi:

  if pyck then
    for i to 7 do
      if not contain( rhs( pyva[ i ] ), rhs( va[ i ] ) ) then break fi
    od:
    if i > 7 then
      checkprocess := [ op( checkprocess ), [ 0 ] ]:
      return true:
    fi:
    for i to 7 do
      if rhs( va[ i ] )[ 2 ] < rhs( pyva[ i ] )[ 1 ]
      or rhs( va[ i ] )[ 1 ] > rhs( pyva[ i ] )[ 2 ] then
        break
      fi
    od:
    if i <= 7 then
      pyck := false
    fi:
  fi:

  for i from 1 to nops( methods ) do
    # check the domain through specific in turn,
    # the order arranged for the methods may influence efficiencies
    if ischecked( va, method = methods[ i ] ) then
      checkprocess := [ op( checkprocess ), [ i ] ]:
      return true:
    fi:
  od:

  # subdivide the domain over the widest interval
  dim := maxwidthdim( va ):
  # when the critical value is set to 1/1000,
  # domain1_1105_1101_1100_1099 cannot be check successfully.
  if intwidth( va[ dim ] ) < ( 1/10000 ) then
    notchecked := [ op( notchecked ), va ]:
    checkprocess := [ op( checkprocess ), [ -4 ] ]:
    return false:
  fi:
  # record the position of subdivision
  checkprocess := [ op( checkprocess ), dim ]:
  subint := intsbdv( va, 2, dim ):
  cur := true:
  for i from 1 to nops( subint ) do
    # check subdomains divided
    if procname( subint[ i ], methods,
    `if`( bpyck, NULL, notcheckbipyramid ),
    `if`( pyck, NULL, notcheckpyramid ) ) = false then
      cur := false:
    fi:
  od:
  checkprocess := [ op( checkprocess ), [ `if`( cur, -2, -3 ) ] ]:
  return cur:
end:

# generate a list which indicates the process of checking
subdivisionprocess := proc( )
  local t, t1, x, y:
  global sp:
  t := [ $1..2187 ]:
  t1 := x -> op( map( y -> y = t, x ) ):
  sp := t:
  sp := subsop( t1( [ 62, 158, 239, 863, 1102, 1105, 1106, 2114, 2132] ),
  sp ):
  sp := subsop( t1( [ [ 1105,1101 ], [ 1106,861 ], [ 1106,834 ],
  [ 1106,1099 ], [1106,1100 ] ] ), sp ):
  sp := subsop( t1( [ [ 1105, 1101, 1100], [1106, 834, 725],
  [1106, 834, 726]]), sp ):
  sp := subsop( t1( [ [ 1106,834,725,1752 ], [ 1106,834,726,1507 ],
  [ 1106, 834, 726, 1750 ] ] ), sp ):
end:

# find the subdomain subdivided in the checking process
# corresponding to an integer or a list integer in varsrange1
subdivisiondomain := proc( dmcur::Or( integer, list( integer ) ) )
  local d, indm2, i, t, v:
  uses IntervalArithmetic:
  d := [ ]:
  indm2 := false:
  for i from 2 to nargs do
    if type( args[ i ], list ) then
      d := args[ i ]
    elif type( args[ i ], identical(dm2) ) then
      indm2 := true
    else
      error "invalid argument: 
    fi
  od:
  if d = [ ] then
    t := `if`(indm2, varsrange2, varsrange1 ):
    d := map( rhs, t ):
   fi:

  if type( dmcur, integer ) then
    v := Vector( 7, convert( dmcur-1, base, 3 ) ):
    [ seq(
    [ d[i][1]+intwidth(d[i])/3*v[i], d[i][1]+intwidth(d[i])/3*(v[i]+1) ],
    i=1..7 ) ]
  elif type( dmcur, list ) and nops( dmcur ) = 1 then
    procname( op( dmcur ), args[2..-1] )
  elif dmcur = [ ] then
    t := `if`(indm2, varsrange2, varsrange1 ):
    map( rhs, t )
  else
    procname( dmcur[ -1 ], procname( dmcur[1..-2], args[2..-1] ),
    args[2..-1] )
  fi
end:

# determine the subdomain in which a point or a domain is contained
inwhichdomain := proc( ls::list )
  local b, indm2, d, i, v, m, t:
  global sp:
  uses IntervalArithmetic:
  if Digits < 100 then Digits := 100 fi:
  b := [ ]: indm2 := false:
  for i from 2 to nargs do
    if type( args[ i ], list ) then
      b := args[ i ]
    elif type( args[ i ], identical(dm2) ) then
      indm2 := true
    else
      error "invalid argument: 
    fi
  od:
  # subdomain corresponding the list b
  d := subdivisiondomain( b, `if`( indm2, 'dm2', NULL ) ):

  if type( ls, list( Or( rational, 'interval') ) ) then
    if member( false, zip( contain, d, ls ) ) then
      print( `the point/domain input is not in any domain subdivided in the \
checking process` ):
      return FAIL
    fi:
    v := [seq(
    floor( ( `if`( type( ls[ i ], 'interval' ),
    ls[i][1], ls[i] )-d[i][1] )*3/intwidth( d[i]) ),
    i=1..7 ) ]:
    # when ls[i] = d[i,2]
    v := map( t -> `if`( t = 3, 2, t ), v ):
    m := add( v[ i ]*3^( i-1 ), i = 1..7 )+1:

    # When the domain is in varsrange2, there is only once subdivision.
    if indm2 then
      #print( cat( `the point/domain input is in:`,
      #cat( domain2, _, op( map( t -> (t,_), b ) ), m ) ) ):
      return [ domain2, op( b ), m ]
    fi:

    if op( [ op( b ), m ], sp ) = m then
      #print( cat( `the point/domain input is in:`,
      #cat( domain1, _, op( map( t -> (t,_), b ) ), m ) ) ):
      return [ domain1, op( b ), m ]
    else
      procname( ls, [ op( b ), m ] )
    fi:
  elif type( ls, list( constant ) ) then
    procname( map( `Evalr/shake`, ls ), b )
  elif type( ls[ 1 ], `=` ) then
    procname( map( rhs, ls ), b )
  else
    error "invalid argument: 
  fi:
end:

# search the stored data to determine how a point can be checked
howchecked := proc(ls::list)
  local x, y, indm2, methods, t, fn, pd, currentdomain, i, cur, lev:
  uses IntervalArithmetic:

  if type(ls[1], `=`) then
    return procname( map( rhs, ls ) )
  elif member(true, map( t->not type( t,rational ), ls )) then
    return procname( convert( evalf( ls ), ft2rat ) )
  fi:

  # first suppose that ls is in varsrange1
  indm2 := false:

  # input coordinates reprensent a distribution close to the bipyramid
  if not member( false, zip( contain, bpyva, ls ) ) then
    print( cat( `via the "negative definite" interval matrix method, we had \
first exclude a rectangular region whose midpoint correspond to the \
bipyramid configuration, and radius is `, convert( eta, name ), `, and this \
point just belongs to the region.` ) ):
    return

  # input coordinates reprensent a distribution close to the pyramid
  elif not member( false, zip( contain, pyva, ls ) ) then
    print( cat( `via the "nonnegative definite" interval matrix method, we \
had first exclude a rectangular region whose midpoint correspond to the \
pyramid configuration, and radius is `, convert( zeta, name ), `, and this \
point just belongs to the region.` ) ):
    return

  # input coordinates are not in the variable ranges checked
  elif member( false, zip( contain, varsrange1, ls ) ) and
  member( false, zip( contain, varsrange2, ls ) ) then
    print(`the point is not in the checked range(-Pi/2<=phi<=Pi/2 and \
-Pi<=theta<=Pi)`):
    return

  # ls is contained in varsrange2.
  elif not member( false, zip( contain, varsrange2, ls ) ) then
    indm2 := true
  fi:

  # methods used in the checking process
  methods := convert([pointc, halfsphere, mindis, secondlength, totaldis,
  gradient, derivative],table):
  methods[-1] := bipyramid: methods[0] := pyramid:
  methods[-4] := toosamlldomain: methods[-3] := fail:
  methods[-2] := success:
  # check which domain is ls contained in
  fn := inwhichdomain(ls, `if`( indm2, 'dm2', NULL ) ):
  currentdomain := subdivisiondomain( fn[2..-1], `if`( indm2, 'dm2', NULL )):
  # input the path of the directory which stores checking data
  pd := readstat( "Enter the checking files directory:" ):
  if pd = NULL then
    currentdir( "D:/program/maple/paper/5 points on a sphere/data" )
  else
    currentdir( pd )
  fi:
  # read a stored checking file
  read cat( op( map( t -> (t,`/`),
  [ seq( cat( op( map( t -> (t,_), fn[1..i] )[1..-2] ) ),
  i=1..nops(fn)-1)] ) ), cat( op( map( t -> (t,_), fn )[1..-2] ), `.txt` ) ):
  # print information of the file
  print( `the point input is(expressed in float degree): ` ):
  print( evalf( ls*180/Pi ) ):
  print( cat( `it is in: `, \
cat( `if`( indm2, domain2, domain1 ), _, op( map( t -> (t,_), fn[2..-1] )\
[1..-2] ) ) ) ):
  print( cat( `it was used: `, convert( tm, name ), \
` seconds to check this domain` ) ):

  cur := true: lev := 0:
  for i to nops( checkprocess ) do
    if cur and type( checkprocess[i], integer ) then
      if ls[ checkprocess[ i ] ] <=
      ( currentdomain[ checkprocess[ i ] ][1] +
      currentdomain[ checkprocess[ i ] ][2] )/2 then
        t := 1:
      else
        t := 2: cur := false:
      fi:
      currentdomain := intsbdv( currentdomain, 2, checkprocess[ i ] )[ t ]:
    elif cur then
      print( `it is finally contained in domain(expressed in radian): ` ):
      print( currentdomain ):
      print( `that is(expressed in float degree): ` ):
      #print( `it is finally contained in domain(expressed in float \
degree): ` ):
      print( evalf( currentdomain*180/Pi ) ):
      print( cat( `this domain is checked by method: `, \
methods[ op( checkprocess[ i ] ) ] ) ):
      break:
    else
      if type( checkprocess[ i ], integer ) then lev := lev+1
      elif checkprocess[ i ] in [ [-3], [-2] ] then lev := lev-1
      fi:
      if lev = 0 then cur := true fi:
    fi:
  od:
  return
end:

end module:
\end{maplettyout}
\clearpage

\end{appendix}
\clearpage

\end{document}